\documentclass[a4paper,12pt]{amsart}

\usepackage[lmargin=2.7cm,rmargin=2.7cm,bmargin=3.2cm,tmargin=3cm]{geometry}

\usepackage[utf8]{inputenc}
\usepackage{amsmath}
\usepackage{amssymb}
\usepackage{amsthm}
\usepackage{textcomp}
\usepackage{enumitem}
\usepackage[sort,numbers]{natbib}
\usepackage{array}
\usepackage{bbm}
\usepackage{verbatim}
\usepackage{hyperref}
\usepackage{color} 
\usepackage{graphicx}
\usepackage{booktabs}
\usepackage{multirow}
\usepackage[normalem]{ulem}
\usepackage{threeparttable}
\newcommand{\defeq}{\mathrel{\mathop:}=}

\newcommand{\supp}{\mathop{\mathrm{supp}}}

\newcommand{\X}{\mathsf{X}}
\newcommand{\T}{\mathsf{T}}
\newcommand{\Y}{\mathsf{Y}}

\newcommand{\Ls}{\mathcal{L}}

\newcommand{\Sp}{\mathsf{S}}
\newcommand{\cX}{\mathcal{X}}

\newcommand{\cSp}{\mathcal{S}}

\newcommand{\simiid}{\overset{\mathrm{i.i.d.}}{\sim}}
\newcommand{\pr}{\mathrm{pr}}

\newcommand{\bigmid}{\mathrel{\big|}}

\newcommand{\biggmid}{\mathrel{\bigg|}}
\newcommand{\uarg}{\,\cdot\,}
\newcommand{\ud}{\mathrm{d}}
\newcommand{\R}{\mathbb{R}}
\newcommand{\N}{\mathbb{N}}
\renewcommand{\P}{\mathbb{P}}
\newcommand{\E}{\mathbb{E}}

\newcommand{\charfun}[1]{\mathbb{I}\left(#1\right)}
\newcommand{\unitfun}{\mathbf{1}}

\newcommand{\given}{\,:\,}

\newcommand{\Var}{\mathrm{Var}}

\newcommand{\param}{\theta}
\newcommand{\Param}{\Theta}
\newcommand{\paramspace}{\T}
\newcommand{\latent}{x}
\newcommand{\Latent}{X}
\newcommand{\latentspace}{\X}
\newcommand{\obs}{y}
\newcommand{\Obs}{Y}
\newcommand{\obsspace}{\Y}
\newcommand{\margdens}{\pi_m}
\newcommand{\approxdens}{\pi_a}
\newcommand{\jointdens}{\pi}
\newcommand{\conddens}{r}
\newcommand{\ssmhid}{z}
\newcommand{\Ssmhid}{Z}
\newcommand{\ssmobs}{y}
\newcommand{\Ssmobs}{Y}
\newcommand{\properweight}{W}
\newcommand{\proper}{\xi}
\newcommand{\properweightalt}{V}
\newcommand{\properalt}{\zeta}
\newcommand{\meanfunc}{\mu}

\newtheorem{theorem}{Theorem}
\newtheorem{corollary}[theorem]{Corollary}
\newtheorem{proposition}[theorem]{Proposition}
\newtheorem{lemma}[theorem]{Lemma}

\theoremstyle{remark}
\newtheorem{remark2}[theorem]{Remark}

\theoremstyle{definition}
\newtheorem{definition2}{Definition}
\newtheorem{assumption}{Assumption}
\newtheorem{algorithm}{Algorithm}

\begin{document}

%%%%%%%%%%%%%%%%%%%%%%%%%%%%%%%%%%%%%%%%%%%%%%%%%%%%%%%%%%%%%%%%%%%%%%%%%%%%%%%
% Title, authors, keywords %{{{
\title[Importance sampling type estimators based on
  approximate MCMC]{Importance sampling type estimators based on
    approximate marginal MCMC}
\author{Matti Vihola}
\address{University of Jyväskylä, Department of Mathematics and Statistics, 
P.O.Box 35, FI-40014 University of Jyväskylä, Finland}
\email[Matti Vihola]{matti.vihola@iki.fi}

\author{Jouni Helske}
\address{University of Jyväskylä, Department of Mathematics and Statistics, 
P.O.Box 35, FI-40014 University of Jyväskylä, Finland}
  
\author{Jordan Franks}
\address{Newcastle University,
  School of Mathematics, Statistics and Physics,
NE1 7RU Newcastle-upon-Tyne, United Kingdom}

%}}}

%%%%%%%%%%%%%%%%%%%%%%%%%%%%%%%%%%%%%%%%%%%%%%%%%%%%%%%%%%%%%%%%%%%%%%%%%%%%%%
\begin{abstract} %{{{
We consider importance sampling (IS) type weighted estimators 
based on Markov chain Monte Carlo (MCMC) targeting an approximate
marginal of the target distribution. In the context of Bayesian latent
variable models, the MCMC typically operates on the hyperparameters,
and the subsequent weighting may be based on IS or
sequential Monte Carlo (SMC), but allows for multilevel techniques as
well. The IS approach provides a natural alternative to delayed
acceptance (DA) pseudo-marginal/particle MCMC, and has many
advantages over DA, including a straightforward parallelisation and
additional flexibility in MCMC implementation. We detail minimal
conditions which ensure strong consistency of the suggested
estimators, and provide central limit theorems with expressions for
asymptotic variances. We demonstrate how our method can make use of SMC
in the state space models context, using Laplace
approximations and time-discretised diffusions. Our experimental
results are promising and show that the IS type approach can provide
substantial gains relative to an analogous DA scheme, and is often
competitive even without parallelisation.
\end{abstract} %}}}

\maketitle

%%%%%%%%%%%%%%%%%%%%%%%%%%%%%%%%%%%%%%%%%%%%%%%%%%%%%%%%%%%%%%%%%%%%%%%%
\section{Introduction} %{{{

Markov chain Monte Carlo (MCMC) has become a standard tool in Bayesian
analysis. The greatest benefit of MCMC is its general applicability
--- it is guaranteed to be consistent with virtually no assumptions on
the underlying model. However, the practical applicability of MCMC
generally depends on the dimension of the unknown variables, the
number of data, and the computational resources available. Because
MCMC is only asymptotically unbiased, and sequential in nature, it can
be difficult to implement efficiently with modern parallel and
distributed computing facilities
\cite{lee-yau-giles-doucet-holmes,green-latuszynski-pereyra-robert,wilkinson}.

We promote a simple two-phase inference approach, based on importance
sampling (IS), which is well-suited
for parallel implementation. It combines a typically low-dimensional
MCMC targeting an approximate marginal distribution with independently
calculated estimators, which yield exact inference over the full
posterior. The estimator is similar to self-normalised importance
sampling, but is more general, allowing for sequential Monte
Carlo and multilevel type corrections. The method is naturally
applicable in a latent variable models context, where the MCMC operates
on the hyperparameter distribution using an approximate marginal
likelihood, and re-weighting is based on a sampling scheme on the
latent variables. We detail the application of the method with
Bayesian state space models, where we use importance sampling and
particle filters for correction. 

%%%%%%%%%%%%%%%%%%%%%%%%%%%%%%%%%%%%%%%%
\subsection{Related work} %{{{

We consider a framework which combines and generalises upon various previously
suggested methods, which, to our knowledge, has not been
systematically explored before. Importance sampling correction of MCMC
has been suggested early in the MCMC literature
\cite[e.g.][]{hastings,glynn-iglehart,doss}, and used, for instance,
to estimate Bayes factors using a single MCMC output
\cite{doss-bayes-factors}. Related confidence intervals have been
suggested based on regeneration \cite{bhattacharya} and in case of
multiple Markov chains \cite{tan-doss-hobert}. Using unbiased
estimators of importance weights in this context has been suggested at
least in
\cite{lin-liu-sloan,lyne-girolami-atchade-strathmann-simpson}, who
consider marginal inference with a generalisation of the
pseudo-marginal method, allowing for likelihood estimators that may
take negative values, and in \cite{quiroz-villani-kohn} with data
sub-sampling. 

Nested or compound sampling has also appeared in many forms in the
Monte Carlo literature. The SMC\textsuperscript{2} algorithm
\cite{chopin-jacob-papaspiliopoulos} is based on an application of
nested sequential Monte Carlo steps, which has similarities with our
framework, and the IS\textsuperscript{2} method
\cite{tran-scharth-pitt-kohn} focuses on the case where the
preliminary inference is based on independent sampling. We focus on the
MCMC approximation of the marginal distribution, which we believe
often to be easily implementable in practice, also when the marginal
distribution has a non-standard form. The Markov dependence in
the marginal Monte Carlo approximation comes with some extra
theoretical issues, which we address in detail.

Our setting highlights explicitly the connection of IS
type correction and delayed acceptance (DA)
\cite{fox-nicholls,liu-mc,christen-fox},
and recently developed pseudo-marginal type
MCMC \cite{andrieu-roberts,lin-liu-sloan} such as particle MCMC
\cite{andrieu-doucet-holenstein}, grouped independence
Metropolis-Hastings \cite{beaumont}, approximate Bayesian
computation (ABC) MCMC \cite{marjoram-molitor-plagnol-tavare}, the algorithm
for estimation of discretely observed diffusions suggested in
\cite{beskos-papaspiliopoulos-roberts-fearnhead}, and annealed
IS \cite{karagiannis-andrieu,neal-annealed}.
Theoretical advances of pseudo-marginal methods
\cite{andrieu-vihola-pseudo,andrieu-vihola-order,doucet-pitt-deligiannidis-kohn,sherlock-thiery-roberts-rosenthal,andrieu-lee-vihola,chopin-singh,lindsten-douc-moulines}
have already led to more efficient implementation of such methods, but
have also revealed fundamental limitations. For instance, the methods
may suffer from slow (non-geometric) convergence in practically
interesting scenarios \cite{andrieu-roberts,lee-latuszynski}.  Adding
dependence to the estimators \cite[cf.][]{andrieu-vihola-order}, such
as using the recently proposed correlated version of the
pseudo-marginal MCMC \cite{deligiannidis-doucet-pitt-kohn}, may help
in more efficient implementation in certain scenarios, but
a successful implementation
of such a method may not always be possible, and the question of
efficient parallelisability remains a challenge. The blocked
parallelisable particle Gibbs \cite{singh-lindsten-moulines} has
appealing limiting properties, but its implementation still requires
synchronisation between every update cycle, which may be costly in
some computing environments.

The IS approach which we propose may assuage some of the
aforementioned challenges of the pseudo-marginal framework, by
replacing a pseudo-marginal MCMC with a fast-mixing but approximate
MCMC, and postponing the computationally intensive calculations to
parallelisable post-processing; see Section \ref{sec:algorithmic-da}.

%}}}

%%%%%%%%%%%%%%%%%%%%%%%%%%%%%%%%%%%%%%%%
\subsection{Outline} %{{{

We introduce a generic Bayesian latent variable model in Section
\ref{sec:latent}, detail our approach algorithmically,
and compare it with DA. We also discuss practical
implications, modifications and possible extensions. After introducing
notation in Section \ref{sec:notation}, we formulate general
IS type correction of MCMC and 
related consistency results in Section \ref{sec:simple}. We detail the
general case (Theorem
\ref{thm:proper-consistency}), based on a concept (Definition
\ref{def:proper}), which we call a `proper weighting' scheme (following
the terminology of Liu \cite{liu-mc}), which is natural and
convenient in many contexts.
In Section \ref{sec:asvar-clt},
we state central limit theorems 
and expressions for asymptotic variances.
Section \ref{sec:block} focuses on estimators
which calculate IS correction once for each accepted state, stemming
from a so-called `jump chain' representation.
Section \ref{sec:pseudo} details consistency of our estimators in case
the approximate chain is pseudo-marginal.

We then focus on state-space models with linear-Gaussian state
dynamics in Section \ref{sec:lin-gauss-ssms}, and show how a Laplace
approximation can be used both for approximate inference, and for
construction of efficient proper weighting schemes. Section
\ref{sec:diffusion-ssms} describes an instance of our approach in the
context of discretely observed diffusions, with an approximate
pseudo-marginal chain. We compare empirically
several algorithmic variations in Section \ref{sec:exp} with Poisson
observations, with a stochastic volatility model and with a discretely
observed geometric Brownian motion. Section
\ref{sec:discussion} concludes, with discussion.

%}}}

%}}}

%%%%%%%%%%%%%%%%%%%%%%%%%%%%%%%%%%%%%%%%%%%%%%%%%%%%%%%%%%%%%%%%%%%%%%%%
\section{The proposed latent variable model inference methodology}
\label{sec:latent} %{{{

%{{{

A generic Bayesian latent variable model is defined
in terms of three
random vector, and corresponding conditional densities:
\begin{itemize}
    \item $\Param \sim \pr(\uarg)$  --- prior density of (hyper)parameters,
    \item $\Latent \mid \Param=\param \sim
      \mu^{(\param)}(\uarg)$ --- prior of latent variables given
      parameters, and
    \item $\Obs \mid (\Param=\param,\Latent=\latent)
      \sim g^{(\param)}(\uarg\mid
      \latent)$  --- the observation model.
\end{itemize}
The aim is inference over the posterior of $(\Param,\Latent)$
given observations $\Obs=\obs$, with density
$\pi(\param,\latent) \propto \pr(\param)
    \mu^{(\param)}(\latent) g^{(\param)}(\obs\mid\latent)$.
Standard MCMC algorithms may, in principle,  be applied directly
for inference, but the typical high dimension of
the latent variable $\latent$ and the common strong dependency
structures often lead to poor performance of generic algorithms.

Our inference approach focuses on the specific structure of the model,
based on the 
factorisation $\jointdens(\param,\latent )
    = \margdens(\param) \conddens(\latent\mid\param)$,
where the marginal posterior density $\margdens$ and the corresponding
conditional $\conddens$ are:
\begin{align*}
    \margdens(\param) \defeq \int \jointdens(\param,\latent)\ud \latent
    \propto \pr(\param) L(\param) \qquad\text{and}\qquad
    \conddens(\latent\mid \param) \defeq
    \frac{p^{(\param)}(\latent,\obs)}{L(\param)},
\end{align*}
with the joint density of the latent and the observed $p^{(\param)}(\latent, \obs)$,
and the marginal likelihood $L(\param)$ given as follows:
\begin{align*}
    p^{(\param)}(\latent, \obs) &\defeq
    \mu^{(\param)}(\latent) g^{(\param)}(\obs\mid
    \latent) &\qquad&\text{and}\qquad&
   L(\param) &\defeq \int
    p^{(\param)}(\latent,\obs) \ud \latent.
\end{align*}
Two particularly successful latent variable model inference methods,
the integrated nested Laplace approximation (INLA) \cite{rue-martino-chopin}
and the particle MCMC methods (PMCMC)
\cite{andrieu-doucet-holenstein}, rely on this structure.
In essence, the INLA is based on an efficient Laplace approximation
$p_a^{(\param)}(\latent,\obs)$ of $p^{(\param)}(\latent,\obs)$,
determining an approximate marginal likelihood $L_a(\param)$
and approximate conditional distribution
$\conddens_a(\latent\mid\obs)$.
Particle MCMC uses a specialised SMC algorithm,
which provides an unbiased approximation of
expectations with respect to 
$p^{(\param)}(\latent,\obs)$ allowing for exact inference, and which is
particularly efficient in the state space models context.

%}}}

%%%%%%%%%%%%%%%%%%%%%%%%%%%%%%%%%%%%%%%%
\subsection{An algorithmic description}
\label{sec:algorithmic-method} %{{{

The primary aim of this paper is the efficient use of an approximate
marginal likelihood $L_a(\param)$ within a Monte Carlo framework that
leads to efficient, parallelisable and exact inference. For instance,
Laplace approximations often lead to a natural choice for
$L_a(\param)$. The inference method
which we propose comprises two algorithmic phases, which are
summarised below:
\begin{enumerate}
 \item[Phase 1:]
%   \label{item:approx-phase}
   Simulate a Markov chain $(\Param_k)_{k=1,\ldots,n}$ targeting
   an approximate hyperparameter posterior
   $
       \approxdens(\param) \propto \pr(\param) L_a(\param).
   $
 \item[Phase 2:]
%   \label{item:is-phase}
   For each $\Param_k$, sample
   $(\properweightalt_k^{(i)},\Latent_k^{(i)})_{i=1,\ldots,m}$
   where $V_k^{(i)}\in\R$ and $\Latent_k^{(i)}$ are in the latent
   variable space,
   and calculate $\properweight_k^{(i)} \defeq
   \properweightalt_k^{(i)}/L_a(\Param_k)$,
   which determine a weighted estimator
   \begin{equation}
       E_n(f) \defeq
       \frac{\sum_{k=1}^n \sum_{i=1}^m \properweight_k^{(i)}
         f(\Param_k, \Latent_k^{(i)})}{\sum_{j=1}^n \sum_{\ell=1}^m
         \properweight_j^{(\ell)}}
       \label{eq:corrected-estimator}
   \end{equation}
   of the full posterior expectation $ \E_\jointdens[f(\Param,\Latent)]
   = \int
   f(\param,\latent)\jointdens(\param,\latent)\ud \param\ud\latent$.
\end{enumerate}
The following conditions are essential to ensure the consistency of the 
estimator:
\begin{enumerate}
    \item[C1:]
      The approximation is consistent, in the sense that 
      $L_a(\param)>0$ whenever
      $L(\param)>0$, and $\int \pr(\param)L_a(\param)\ud \param<\infty$.
    \item[C2:]
      The Markov chain $(\Param_k)_{k\ge n}$ is Harris
      ergodic (Definition \ref{def:harris}) with respect to $\approxdens$.
    \item[C3:]
      Denoting $f^*(\param) \defeq \E_{\pi}[f(\Param,\Latent)\mid\Param=\param]
= \int r(\latent\mid\param) f(\param,\latent)\ud\latent$, there exists
a constant $c_w>0$ such that
\begin{align}
    \E\bigg[ \sum_{i=1}^m \properweightalt_k^{(i)}
    f(\Param_k,\Latent_k^{(i)})\biggmid \Param_k=\param\bigg]
    &= c_w L(\param) f^*(\param),
    \label{eq:simple-unbiasedness}
\end{align}
for all $\param\in\paramspace$,
all functions $f$ of interest, and for $f\equiv 1$ 
(i.e.~\eqref{eq:simple-unbiasedness} holds with
$f(\uarg)$ and $f^*(\uarg)$ omitted).
The value of $c_w$ need not be known.
\end{enumerate}
Both C1 and C2  are easily
satisfied by construction of the approximation, and
C3 is satisfied by many schemes. Appendix
\ref{sec:ssm} reviews how the particle filter leads to such schemes. 
The conditions C1--C3 are key, but unfortunately 
not enough to guarantee consistency, which requires 
also a (mild) integrability condition, which
$(\properweight_k^{(i)},\Latent_k^{(i)})_{k=1,\ldots, n;\, i=1,\ldots, m}$ 
must satisfy. Fortunately, in the common case of non-negative
$\properweightalt_k^{(i)}$, this integrability is guaranteed if
also $|f|$ satisfies \eqref{eq:simple-unbiasedness}; 
see Section \ref{sec:simple} for
details. Further conditions ensure a central limit theorem
$\sqrt{n}\{E_n(f) - \E_\jointdens[f(\Param,\Latent)]\}\to
N(0,\sigma^2)$, as detailed in Section \ref{sec:asvar-clt}.

When Phase 1 is a Metropolis-Hastings algorithm,
it is possible to generate only one batch of
$(\tilde{\properweightalt}_k^{(i)},\tilde{\Latent}_k^{(i)})_{i=1,\ldots,m}$
for each \emph{accepted} state $(\tilde{\Param}_k)$. If $N_k$ stands
for the time spent at $\tilde{\Param}_k$, then the corresponding weights are
determined as $\tilde{W}_k \defeq N_k
V_k^{(i)}/L_a(\tilde{\Param}_k)$; see
Section \ref{sec:block} for details about such `jump chain' estimators.

%}}}

%%%%%%%%%%%%%%%%%%%%%%%%%%%%%%%%%%%%%%%%
\subsection{Use with approximate pseudo-marginal MCMC}
\label{sec:algorithmic-pseudo} %{{{

In many scenarios, such as with time-discretised diffusions, the
latent variable prior density $\mu^{(\param)}$ cannot be evaluated,
and exact simulation is impossible or very expensive. Simulation is
also expensive with a fine enough time-discretisation.

A coarsely discretised model leads to a natural cheap approximation
$\hat{\mu}^{(\param)}$, but in Phase 1, the Markov
chain will often be a pseudo-marginal MCMC
\cite[cf.][]{andrieu-roberts}, in which case our scheme would have the
following form:
\begin{enumerate}
 \item[Phase 1':]
%   \label{item:approx-phase2}
   Simulate a pseudo-marginal Metropolis-Hastings chain
   $(\Param_k,U_k)$ for $k=1,\ldots,n$, following
   \begin{enumerate}
       \item Draw a proposal $\tilde{\Param}_k$ from $q(\Param_{k-1},\uarg)$
         and given $\tilde{\Param}_k$, construct an estimator
         $\tilde{U}_k\ge 0$ such that $\E[\tilde{U}_k\mid
         \tilde{\Param}_k=\param] = L_a(\param)$.
       \item With probability
         $
             \min\Big\{1,\frac{\pr(\tilde{\Param}_k)
                 \tilde{U}_k q(\tilde{\Param}_k,\Param_{k-1})
                 }{
                 \pr(\Param_{k-1})U_{k-1}
               q(\Param_{k-1},\tilde{\Param}_k)}
               \Big\}$,
         accept and set $(\Param_k,U_k) =
         (\tilde{\Param}_k,\tilde{U}_k)$; otherwise reject the move.
   \end{enumerate}
   \item[Phase 2':]
%     \label{item:pseudo-correct-phase2}
     For each $(\Param_k,U_k)$, sample
   $(\properweightalt_k^{(i)},\Latent_k^{(i)})_{i=1,\ldots,m}$
   and set $\properweight_k^{(i)} \defeq \properweightalt_k^{(i)}/U_k$,
   which determine the estimator as in \eqref{eq:corrected-estimator}.
\end{enumerate}
Algorithmically, the pseudo-marginal version above is similar to the
method in Section \ref{sec:algorithmic-method}, with the likelihood
$L_a(\Param_k)$ replaced with its estimator $U_k$. The requirements
for the approximate likelihood C1
and its estimator C3 remain identical, and
C2 must hold for the pseudo-marginal chain
$(\Param_k,U_k)_{k\ge 1}$, together with the following condition:
\begin{enumerate}
    \item[C4:]
    The estimators $\tilde{U}_k$ are strictly positive, almost surely,
    for all $\tilde{\Param}_k\in\paramspace$.
\end{enumerate}
These are enough to guarantee consistency; see
Section \ref{sec:pseudo}, and in particular
Proposition \ref{prop:pseudo-approx-is} for details, which also justifies
why C4 is
needed for consistency.
In practice it may be easily
satisfied, because the likelihood estimators $\tilde{U}_k$ may be 
inflated, if necessary (see Section \ref{sec:discussion}).

Note that the variables
$(\properweightalt_k^{(i)},\Latent_k^{(i)})_{i=1,\ldots,m}$
may depend on both $\Param_k$ and the related likelihood estimate $U_k$. The
dependency may be useful, if positively correlated $\properweightalt_k^{(i)}$ and
$U_k$ are available, leading to lower variance weights
$\properweight_k^{(i)}=\properweightalt_k^{(i)}/U_k$. This is similar
to the correlated pseudo-marginal algorithm
\cite{deligiannidis-doucet-pitt-kohn}, which relies on a particular form
of $\properweightalt_k^{(i)}$ and $U_k$. 
If positively correlated structure is 
unavailable, $(\properweightalt_k^{(i)},\Latent_k^{(i)})_{i=1,\ldots,m}$
may be constructed independently of $U_k$.

%}}}

%%%%%%%%%%%%%%%%%%%%%%%%%%%%%%%%%%%%%%%%
\subsection{Comparison with delayed acceptance}
\label{sec:algorithmic-da} %{{{

The key condition, under which we believe our method to be useful, is
that the Phase 1 Markov chain is computationally
relatively cheap compared to construction of the random variables
$(\properweight_k^{(i)},\Latent_k^{(i)})$ computed in Phase 2.
Similar rationale, and similar building blocks --- a $\approxdens$-reversible
Markov chain and random variables analogous to
$(\properweight_k^{(i)},\Latent_k^{(i)})$ --- have been suggested
earlier for construction of a delayed acceptance (DA) pseudo-marginal
MCMC scheme \cite[cf.][]{golightly-henderson-sherlock}. Such an
algorithm defines a Markov chain
$(\Param_k,\properweight_k^{(i)},\Latent_k^{(i)})_{k\ge 1}$, with one
iteration consisting of the following steps:
\begin{enumerate}
\item[DA 1:] %\label{item:da-proposal}
  Draw $\tilde{\Param}_k\sim P(\Param_{k-1},\uarg)$.
  If $\tilde{\Param}_k = \Param_{k-1}$ reject and set 
  $(\Param_k, W_k^{(i)}, X_k^{(i)}) = (\Param_{k-1}, W_{k-1}^{(i)},
  X_{k-1}^{(i)})$; otherwise go to
  (DA 2).
\item[DA 2:] %\label{item:da-step}
  Conditional on $\tilde{\Param}_k$, draw
  $(\tilde{\properweightalt}_k^{(i)},\tilde{\Latent}_k^{(i)})_{i=1,\ldots,m}$
  which satisfy
  \eqref{eq:simple-unbiasedness} with $\tilde{\Param}_k$ in place of
  $\Param_k$, and set $\tilde{\properweight}_k^{(i)} \defeq
  \tilde{\properweightalt}_k^{(i)}/L_a(\tilde{\Param}_k)$.
With probability
  $ \min\Big\{1,\frac{\sum_{i=1}^m
          \tilde{\properweight}_k^{(i)}}{\sum_{\ell=1}^m
          \properweight_{k-1}^{(\ell)}}\Big\}$,
  accept and set 
  $(\Param_k,\properweight_k^{(i)},\Latent_k^{(i)}) = 
  (\tilde{\Param}_k,\tilde{\properweight}_k^{(i)},\tilde{\Latent}_k^{(i)})$;
  otherwise reject and set 
  $(\Param_k, W_k^{(i)}, X_k^{(i)}) = (\Param_{k-1}, W_{k-1}^{(i)},
  X_{k-1}^{(i)})$
\end{enumerate}
If the pseudo-marginal method is used in DA 1 
the value $L_a(\Param_k)$ is replaced with the related likelihood
estimator. Under the same assumptions as required by our scheme,
and additionally requiring that $\tilde{\properweight}_k^{(i)}\ge 0$,
and that the Markov chain $(\Param_k,\properweight_k^{(i)},\Latent_k^{(i)})_{k\ge 1}$
is Harris, the DA scheme leads to a consistent estimator:
\[
    \frac{1}{n}\sum_{k=1}^n \sum_{i=1}^m
    \bigg(\frac{\properweight_k^{(i)}}{\sum_{\ell=1}^m
      \properweight_k^{(\ell)}}\bigg)
    f(\Param_k,\Latent_k^{(i)})
    \xrightarrow{n\to\infty} \E_\pi[f(\Param,\Latent)].
\]

Our IS scheme is a natural alternative to such a DA scheme,
replacing the independent Metropolis-Hastings type accept-reject step
DA 2 with analogous weighting. This relatively small
algorithmic change brings many, potentially substantial, benefits over
DA, which we note next.
\begin{enumerate}
\item Phase 2
  corrections are entirely independent ‘post-processing’ of
  Phase 1 
  MCMC output $(\Param_k)_{k=1,\ldots,n}$, which is easy to
  implement efficiently using parallel or distributed computing.
  This is unlike DA 1 and DA 2,
  which must be iterated sequentially.
\item 
  \label{item:is-da-comparison} If Phase 2
  correction variables are calculated only once for each accepted
  $\Param_k$ (so-called `jump chain' representation, see Section
  \ref{sec:block}), the IS method will typically be computationally less
  expensive than DA with the same number of iterations, even without
  parallelisation.
\item 
  \label{item:is-da-comparison-thinning}
  The Phase 1
  MCMC chain $(\Param_k)$ may be (further) thinned before applying
  (much more computationally demanding) Phase 2.
  Thinning of the DA chain is less likely beneficial
  \cite[cf.][]{owen-thinning}.
\item In case the approximate marginal MCMC $(\Param_k)$ is based on
  a deterministic likelihood approximation, it is generally `safer'
  than (pseudo-marginal) DA using likelihood
  estimators, because pseudo-marginal MCMC may have issues with mixing
  \cite[cf.][]{andrieu-vihola-pseudo}. It is also easier to implement
  efficiently. For instance, popular adaptive MCMC methods which rely on acceptance
  rate optimisation \cite[and references therein]{andrieu-thoms}
  are directly applicable. 
\item Reversibility of the MCMC kernel $P$ in
  DA 1 is necessary,
  but not required for the Phase 1 MCMC.
\item The average of the weights, $n^{-1} \sum_{k=1}^n \sum_{i=1}^m
  W_k^{(i)}$, provides a consistent estimator of the ratio of the 
  normalising constants of $\pi$
  and $\pi_a$, which may be useful in some contexts.
\item Non-negativity of $\properweight_k^{(i)}$ is
  required in DA 2, but
  not in Phase 2.
  This may be useful in certain
  contexts, where multilevel \cite{heinrich,giles-or}
  or debiasing \cite{mcleish,rhee-glynn,vihola-unbiased} are
  applicable.
  (See also the discussion in \cite{jacob-thiery} why pseudo-marginal
  method may not be applicable at all in such a context.)
\item The separation of `approximate' Phase 1
  and `exact' Phase  2 allows for two-level
  inference. In statistical practice, preliminary analysis could be based on
  (fast) purely approximate inference, and the (computationally
  demanding) exact method
  could be applied only as a final verification to
  ensure that the approximation did not affect the findings.
\end{enumerate}
To elaborate the last point, the approximate likelihood $L_a(\param)$
is usually based on an approximation $p_a^{(\param)}(\latent,\obs)$ of the
latent model $p^{(\param)}(\latent,\obs)$. If the approximate model
admits tractable expectations of functions $f$ of interest or exact
simulation,
direct approximate inference is possible, because
  \[
      \frac{1}{n} \sum_{k=1}^n
      f_a^* (\Param_k) \to
      \E_{\tilde{\pi}}[f(\Param,\Latent)],
      \qquad \text{where}\qquad
      f_a^*(\param) \defeq
  \E_{\tilde{\pi}}[f(\Param,\Latent)\mid\Param=\param],
  \]
  with approximate joint posterior $\tilde{\pi}(\param,\latent)\propto
  \pr(\param)p_a^{(\param)}(\latent,\obs)$.
  Then,
  Phase 2 allows for quantification of
  the bias
  $\E_{\tilde{\pi}}[f(\Param,\Latent)]-\E_{\pi}[f(\Param,\Latent)]$,
  and confirmation that both inferences lead to the same conclusions.

Even though IS is likely to bring benefits over DA in many scenarios,
there are some situations where DA might perform better. In fact, 
DA may be used in some scenarios where IS cannot be used at all; 
namely, $P$ may only be reversible with respect to a positive measure
that is not a probability measure, and DA can still be valid.
DA may also be more robust with respect to tail behaviour of $\approxdens$, 
where IS needs more care; see the discussion in Section \ref{sec:discussion}.
The further theoretical work
\cite{franks-vihola} provides upper bounds for ratios of the
asymptotic variances of IS and DA in terms of (lower or upper) 
bounds of the weights $W_k^{(i)}$, 
and includes examples where DA outperforms IS 
and vice versa.

%}}}

%}}}

%%%%%%%%%%%%%%%%%%%%%%%%%%%%%%%%%%%%%%%%%%%%%%%%%%%%%%%%%%%%%%%%%%%%%%%%
\section{Notation and preliminaries}
\label{sec:notation} %{{{

Throughout the paper, we consider general state spaces while using
standard integral notation. If the model at hand is given in terms of 
standard probability densities, the rest of this 
paragraph can be skipped.
Each space $\X$ is assumed to be equipped with a $\sigma$-finite
dominating measure `$\ud x$' on a
$\sigma$-algebra denoted with a corresponding calligraphic
letter, such as $\cX$. Product spaces are equipped with the related
product $\sigma$-algebras and product dominating measures.  If $\X$
is a subset of an Euclidean space $\R^d$, $\ud x$ is taken by default
as the Lebesgue measure and $\cX$ as the Borel subsets of $\X$.
$\R_+$ stands for the non-negative real numbers, and constant unit function
is denoted by $\unitfun$.

If $\nu$ is a probability density on $\X$, we define the support of
$\nu$ as $\supp(\nu) \defeq \{x\in\X\given \nu(x)>0\}$, and the
probability measure corresponding to $\nu$ with the same symbol $\nu(\ud
x) \defeq \nu(x)\ud x$.
If $g:\X\to\R$, we denote $\nu(g)\defeq \int
g(x)\nu(\ud x)$, whenever well-defined.
For a probability density or measure $\nu$ on $\X$ and
$p\in[1,\infty)$, we denote by $L^p(\nu)$ the set of measurable
$g:\X\to\R$ with $\nu(|g|^p)<\infty$, and by $L_0^p(\nu)\defeq \{g\in
L^p(\nu)\given \nu(g)=0\}$ the corresponding set of zero-mean
functions. If $P$ is a Markov
transition probability, we denote the probability measure $(\nu P)(A)
\defeq \int \nu(\ud x) P(x,A)$, and the function $(Pg)(x) \defeq \int
P(x,\ud y) g(y)$. Iterates of transition probabilities are defined
recursively through $P^n(x,A) \defeq \int P(x,\ud y) P^{n-1}(y,A)$ for
$n\ge 1$, where $P^0(y,A) \defeq \charfun{y\in A}$.

We follow the conventions $0/0\defeq 0$ and $\N\defeq \{1,2,\ldots\}$.
For integers $a\le b$, we denote by $a{:}b$ the integers within the interval $[a,b]$. We use this
notation in indexing, so that $x_{a:b} = (x_a,\ldots,x_b)$, $x^{(a:b)}
= (x^{(a)},\ldots,x^{(b)})$. If $a>b$, then $x_{a:b}$ or $x^{(a:b)}$
is void, so that for example $g(x,y_{1:0})$ is interpreted as $g(x)$.  Similarly,
if $i_{1:T}$ is a vector, then $x^{(i_{1:T})} =
(x^{(i_1)},\ldots,x^{(i_T)})$ and
$x_{1:T}^{(i_{1:T})}=(x_1^{(i_1)},\ldots,x_T^{(i_T)})$.
We also use double-indexing, such as
$x_k^{(1:m,1:n)}=(x_k^{(1,1)}$, $\!\ldots\,$, $x_k^{(1,n)}$, $x_k^{(2,1)}$, 
$\!\ldots\,$, $x_k^{(m,n)})$.

Throughout the paper, we assume the underlying MCMC scheme to satisfy
the following standard condition.
%%%%%%%%%%%%%%%%%%%%
\begin{definition2}[Harris ergodicity] 
    \label{def:harris} %{{{
A Markov chain is called \emph{Harris ergodic} with respect to
$\nu$, if it is $\psi$-irreducible, Harris
recurrent and with invariant probability $\nu$.
\end{definition2}
%}}}
Virtually all MCMC schemes are Harris ergodic
\cite[cf.][]{tierney,nummelin-mcmcs}, although in some cases careless
implementation could lead to a non-Harris chain
\cite[cf.][]{roberts-rosenthal-harris}.
Thanks to the Harris assumption, all the limit theorems which we give
hold for any initial distribution of the related Markov chain.

%}}}

%%%%%%%%%%%%%%%%%%%%%%%%%%%%%%%%%%%%%%%%%%%%%%%%%%%%%%%%%%%%%%%%%%%%%%%%
\section{General importance sampling type correction of MCMC}
\label{sec:simple} %{{{

Hereafter, $\approxdens$ is a probability density on $\paramspace$ and represents an
approximation of a probability density $\margdens$ of interest.
The consistency of IS type correction relies on the following mild
assumption.
%%%%%%%%%%%%%%%%%%%%
\begin{assumption}
    \label{a:mcmc-is} %{{{
The Markov chain $(\Param_k)_{k\ge
  1}$ and the density $\approxdens$ satisfy:
\begin{enumerate}[label=(\roman*)]
    \item
    $(\Param_k)_{k\ge 1}$ is Harris ergodic with
    respect to $\approxdens$.
    \item \label{item:support} $\supp(\margdens) \subset \supp(\approxdens)$.
    \item $w_u(\param) \defeq c_w \margdens(\param)/\approxdens(\param)$, where
      $c_w>0$ is a constant.
\end{enumerate}
\end{assumption} %}}}
If Assumption \ref{a:mcmc-is} holds and
it is possible to calculate the unnormalised importance weight
$w_u(\param)$ pointwise, the chain $(\Param_k)_{k\ge
  1}$ can be weighted in order to approximate $\margdens(g)$
for every $g\in L^1(\margdens)$, using (self-normalised)
importance sampling \cite[e.g.][]{glynn-iglehart,doss}
\begin{equation*}
    \frac{\sum_{k=1}^n w_u(\Param_k) g(\Param_k)}{\sum_{j=1}^n
      w_u(\Param_j)}
    = \frac{n^{-1}\sum_{k=1}^n w_u(\Param_k)
      g(\Param_k)}{n^{-1}\sum_{j=1}^n w_u(\Param_j)}
    \xrightarrow{n\to\infty} \frac{\approxdens(w_u g)}{\approxdens(w_u)}
    = \margdens(g)\qquad\text{almost surely},
\end{equation*}
as Harris ergodicity
guarantees the almost sure convergence of
both the numerator and the denominator.

In case $\margdens$ is a marginal density, which we will focus on, both the ratio
$w_u(\param)$ and the function $g$ (which will be a conditional
expectation) are typically intractable. Instead, it is often possible to
construct unbiased estimators, which may be used in order to estimate
the numerator and the denominator, in place of
$w_u(\Param_k)$ and $g(\Param_k)$, under mild conditions.  In order to
formalise such a setting, we give the following generic condition
for ratio estimators, which resemble the IS
correction above. 

%%%%%%%%%%%%%%%%%%%%
\begin{assumption}
    \label{a:super-general} %{{{
Suppose 
Assumption \ref{a:mcmc-is} holds, 
and let $(S_k)_{k\ge 1}$, where $S_k = \big(A_k,B_k\big)\in\R^2$,
be conditionally independent
given $(\Param_k)_{k\ge 1}$, such that
the distribution of $S_k$ depends only on the value of $\Param_k$, and
\begin{enumerate}[label=(\roman*)]
    \item $f_A(\param)\defeq \E[A_k\mid \Param_k=\param]$ satisfies
      $\approxdens(f_A) = c_w \margdens(g)$,
    \item $f_B(\param) \defeq \E[B_k\mid \Param_k=\param]$ satisfies
      $\approxdens(f_B) = c_w$, and
    \item $\approxdens(m^{(1)})<\infty$ where
    $m^{(1)}(\param) \defeq \E\big[ |A_k| + |B_k| \bigmid
    \Param_k = \param\big]$.
\end{enumerate}
\end{assumption} %}}}

We record the following simple statement which guarantees consistency
under Assumption \ref{a:super-general}.
%%%%%%%%%%%%%%%%%%%%
\begin{lemma}
    \label{lem:super-general} %{{{
If Assumption \ref{a:super-general} holds
for some $g\in L^1(\margdens)$, then
\[
    E_n(g) \defeq \frac{\sum_{k=1}^n A_k}{\sum_{j=1}^n B_j}
    \xrightarrow{n\to\infty} \margdens(g)
    \qquad\text{almost surely}.
\]
\end{lemma} %}}}
The proof of Lemma \ref{lem:super-general} follows by
observing that $(\Param_k,S_k)_{k\ge 1}$ is Harris ergodic, where
$S_k=(A_k,B_k)$, and
the functions $h_1(\param,a,b) = a$ and $h_2(\param,a,b)=b$ are
integrable with respect to its invariant distribution
$\check{\pi}(\ud \param\times\ud s) \defeq \approxdens(\ud \param)
Q(\param,\ud s)$,
where $Q(\param,A) \defeq \P(S_k\in A\mid \Param_k=\param)$; see
Lemma \ref{lem:aug-prop} in Appendix \ref{app:aug}.

In the latent variable model discussed in
Section \ref{sec:latent}, the aim is inference over
a joint target density $\jointdens(\param,\latent)
\defeq \margdens(\param) \conddens(\latent\mid \param)$ on an extended state space
$\paramspace\times\latentspace$. For every function $f\in L^1(\jointdens)$, we denote
by $f^*(\param)\defeq \int
\conddens(\latent\mid \param) f(\param,\latent)\ud
\latent$ the conditional
expectation of $f$ given $\param$, so $\jointdens(f) = \margdens(f^*)$.
The following formalises a scheme which satisfies
Assumption \ref{a:super-general}
with $g = f^*$
and therefore guarantees consistency for a class of functions
$f\in \Ls \subset L^1(\jointdens)$.

%%%%%%%%%%%%%%%%%%%%
\begin{definition2}[$\Ls$-Proper weighting scheme]
    \label{def:proper} %{{{
Suppose Assumption \ref{a:mcmc-is} holds, and let $(P_k)_{k\ge 1}$
be conditionally independent given $(\Param_k)_{k\ge 1}$, such that the
distribution of each $P_k= (M_k,\properweight_k^{(1:M_k)},\Latent_k^{(1:M_k)})$
depends only on the value of $\Param_k$, 
where $M_k\in\N$,
$\properweight_k^{(i)}\in\R$ and $\Latent_k^{(i)}\in\latentspace$.
Define for any $f\in L^1(\jointdens)$,
\[
    \proper_k(f) \defeq \sum_{i=1}^{M_k} \properweight_k^{(i)}
    f(\Param_k,\Latent_k^{(i)}).
\]
Let $\Ls\subset L^1(\jointdens)$ be all the functions for which
\begin{enumerate}[label=(\roman*)]
    \item
$\meanfunc_f(\param)\defeq
\E[\proper_k(f)\mid \Param_k=\param]$ satisfies $\pi_a(\meanfunc_f) = c_w
\pi(f)$, and
\label{item:proper-unbiased}
\item
            $\approxdens(m_f^{(1)})<\infty$ where
            $m_f^{(1)}(\param)
              \defeq \E\big[ |\proper_k(f)| \bigmid
                  \Param_k=\param\big]$.
  \label{item:proper-general}
\end{enumerate}
If $\unitfun\in \Ls$, then
$(\properweight_k^{(1:M_k)},\Latent_k^{(1:M_k)})_{k\ge 1}$ or
equivalently $(\proper_k)_{k\ge 1}$,
form a $\Ls$-\emph{proper weighting scheme}.
\end{definition2} %}}}

%%%%%%%%%%%%%%%%%%%%
\begin{remark2} %{{{
Regarding Definition \ref{def:proper}:
\begin{enumerate}[label=(\roman*)]
    \item
      In case of non-negative weights, that is, $\properweight_k^{(i)}\ge 0$ almost
      surely, we have $|\proper_k(\unitfun)| =
      \proper_k(\unitfun)$, so $f\equiv \unitfun\in\Ls$ if and only if
      \ref{item:proper-unbiased} holds for $f\equiv \unitfun$.
      Further, if \ref{item:proper-unbiased} holds for both $f$ and
      $|f|$, then
      \ref{item:proper-general} holds, because
                $|\proper_k(f)| \le \proper_k(|f|)$.
    \item When certain multilevel \cite{heinrich,giles-or} or debiasing methods
      \cite[cf.][]{mcleish,rhee-glynn,glynn-rhee} are applied,
       $\properweight_k^{(i)}$ generally take
      also negative values. In such a case, an extra integrability
      condition is necessary, and we believe \ref{item:proper-general}
      is required for consistency in general.
    \item Note that $\Ls$ is closed under linear operations, that is,
      if $a,b\in\R$ and $f,g\in\Ls$, then
      $af+bg\in\Ls$. This, together with $\Ls$ containing constant
      functions, implies that if $f\in\Ls$, then
      $\bar{f} \defeq
      f - \pi(f)\in\Ls$.
    \item In fact, $\proper_k$ may be interpreted as a
      \emph{random (signed) measure} $\xi_k(\ud \theta, \ud x) = \sum_{i=1}^{M_k}
      W_k^{(i)}\delta_{\Theta_k}(\ud \theta)\delta_{X_k^{(i)}}(\ud x)$. 
      Our results extend also to a generalisation, where 
      $\delta_{X_k^{(i)}}(\ud x)$ is replaced by another measure,
      as long as 
      \ref{item:proper-unbiased} and \ref{item:proper-general}
      hold. For instance, in the context of Rao-Blackwellisation, 
      we could have Gaussian distributions in place of $\delta_{X_k^{(i)}}(\ud x)$.
\end{enumerate}
\end{remark2} %}}}

The following consistency result is a direct consequence of 
Lemma \ref{lem:super-general}.

%%%%%%%%%%%%%%%%%%%%
\begin{theorem}
    \label{thm:proper-consistency} %{{{
          If $(\proper_k)_{k\ge 1}$ form a $\Ls$-proper weighting scheme,
          then the IS type estimator is consistent, that is,
          \begin{equation}
              E_n(f) \defeq \frac{\sum_{k=1}^n
                \proper_k(f) }{\sum_{j=1}^n
                \proper_j(\unitfun)   }
              \xrightarrow{n\to\infty} \jointdens(f),
              \qquad \text{almost surely}.
              \label{eq:proper-estimator}
          \end{equation}
\end{theorem} %}}}

Let us next exemplify a `canonical' setting of a proper weighting
scheme,
stemming from standard
unnormalised importance sampling.
%%%%%%%%%%%%%%%%%%%%
\begin{proposition}
    \label{prop:augmented-is} %{{{
          Suppose  Assumption \ref{a:mcmc-is} holds and
          $q^{(\param)}(\uarg)$ defines a probability density on
          $\latentspace$
          for each $\param\in\paramspace$ and
          $\supp(\jointdens) \subset \{(\param,\latent)\given
            \approxdens(\param)q^{(\param)}(\latent)>0\}$.
          Let
          \[
              \Latent_k^{(1:m)}\simiid q^{(\Param_k)},
              \quad
              \properweightalt_k^{(i)} \defeq
              \frac{1}{m}\cdot
              \frac{c_w\jointdens(\Param_k,
                \Latent_k^{(i)})}{q^{(\Param_k)}(\Latent_k^{(i)})}
              \quad\text{and}\quad
              \properweight_k^{(i)} \defeq
              \frac{\properweightalt_k^{(i)}}{\approxdens(\Param_k)},
          \]
          where $c_w>0$ a constant. Then,
          $(\properweight_k^{(1:m)},\Latent_k^{(1:m)})_{k\ge 1}$ form a
          $L^1(\jointdens)$-proper weighting scheme.
\end{proposition} %}}}
When the weights are all positive, sub-sampling may be used in order
to save memory.
%%%%%%%%%%%%%%%%%%%%
\begin{proposition}
    \label{prop:randomise-proper} %{{{
Suppose that $(\properweight_k^{(1:M_k)}, \Latent_k^{(1:M_k)})_{k\ge 1}$
forms a $\Ls$-proper
weighting scheme with non-negative $\properweight_k^{(1:M_k)}\ge 0$ (a.s.).
Let $\properweight_k \defeq \sum_{i=1}^{M_k} \properweight_k^{(i)}$ and
let $(I_k)$ be random variables conditionally independent of
$(\Param_k,\Latent_k^{(i)})$
such that $\P(I_k = i) = \properweight_k^{(i)}/\properweight_k$ (and let $I_k=1$
if $W_k=0$).
Then, $(W_k, \Latent_k^{(I_k)})_{k\ge 1}$
forms a $\Ls$-proper weighting scheme.
\end{proposition} %}}}
The sub-sampling estimator simplifies to
              \[
                  E_n(f) =
                  \frac{\sum_{k=1}^n W_k  f(\Param_k,\Latent_k^{(I_k)})
                    }{
                    \sum_{k=1}^n W_k }.
              \]
We conclude by a complementary statement about convex
combinations of multiple proper sampling schemes.
%%%%%%%%%%%%%%%%%%%%
\begin{proposition}
    \label{prop:convex-proper} %{{{
Suppose $(\proper_{k,j})_{k\ge 1}$ forms
a $\Ls$-proper weighting scheme for each $j\in\{1{:}N\}$,
then, for any constants
$\beta_{1},\ldots,\beta_N\ge 0$ with $\sum_{j=1}^N
\beta_j = 1$, the convex combinations
$\proper_k(f) \defeq \sum_{j=1}^N \beta_j \proper_{k,j}(f)$
form a $\Ls$-proper sampling scheme.
\end{proposition} %}}}

%}}}

%%%%%%%%%%%%%%%%%%%%%%%%%%%%%%%%%%%%%%%%%%%%%%%%%%%%%%%%%%%%%%%%%%%%%%%%
\section{Asymptotic variance and a central limit theorem}
\label{sec:asvar-clt} %{{{

The asymptotic variance is a common efficiency
measure for Markov chain estimators, because it
coincides with the limiting variance of a central limit theorem (CLT),
under general conditions \citep{kipnis-varadhan,maxwell-woodroofe}.

%%%%%%%%%%%%%%%%%%%%
\begin{definition2}
\label{def:asvar} %{{{
Suppose the Markov chain $(\Param_k)_{k\ge 1}$ on $\paramspace$ has
transition probability $P$ which is Harris ergodic with respect to
invariant probability $\approxdens$. For
$f\in L^2(\approxdens)$, the asymptotic variance of $f$ with respect to $P$ is
\[
    \Var(f,P)  \defeq \lim_{n\to\infty} \E\bigg(
    \frac{1}{\sqrt{n}}\sum_{k=1}^n \big[f(\Param_k^{(s)}) - \approxdens(f)\big]\bigg)^2,
\]
whenever the limit exists in $[0,\infty]$,
where $(\Param_k^{(s)})_{k\ge 1}$ stands for the \emph{stationary Markov chain}
with transition
probability $P$, that is, with $\Param_1^{(s)}\sim \approxdens$.
\end{definition2} %}}}

In what follows, we denote by $\bar{f}(\param,\latent) =
f(\param,\latent) - \jointdens(f)$
the centred version of any $f\in L^1(\jointdens)$, and
recall that if $f\in\Ls$, then $\bar{f}\in\Ls$. We also denote
$m_f^{(2)}(\param) \defeq \E[|\proper_k(f)|^2\mid\Param_k=\param]$
for any $f\in\Ls$. 
The proof of the following CLT is given in Appendix
\ref{app:clt}.

%%%%%%%%%%%%%%%%%%%%
\begin{theorem}
\label{thm:proper-clt-rev} %{{{
      Suppose that the
      conditions of Theorem \ref{thm:proper-consistency} are
      satisfied, and $(\Param_k)_{k\ge 1}$ is aperiodic. Let $f\in
      \Ls\cap L^2(\jointdens)$ and
      denote $\bar{f}(\param,\latent) \defeq f(\param,\latent) - \jointdens(f)$.
      If $\approxdens(m_{\bar{f}}^{(2)})<\infty$ and
      either of the following hold:
      \begin{enumerate}[label=(\roman*)]
          \item
            \label{item:proper-kv}
            $(\Param_k)_{k\ge 1}$ is reversible and
            $\Var(\meanfunc_{\bar{f}}, P)<\infty$, or
          \item
            \label{item:proper-mw}
            $\sum_{n=1}^\infty n^{-3/2}
            \big\{\margdens\big(
              \big[\sum_{k=0}^{n-1} P^{k} \meanfunc_{\bar{f}}
              \big]^2\big)\big\}^{1/2}<\infty$,
      \end{enumerate}
      then, the estimator $E_n(f)$ defined in
      \eqref{eq:proper-estimator} satisfies a CLT:
      \begin{equation}
              \sqrt{n}[E_n(f)-\jointdens(f)] \xrightarrow{n\to\infty}
              N\big(0,\sigma_f^2\big),\quad\text{in distribution, where}\quad
              \sigma_f^2\defeq \frac{\Var(\meanfunc_{\bar{f}}, P) +
                \approxdens(v)}{c_w^2},
              \label{eq:is-asvar}
       \end{equation}
      and where
      $v(\param) \defeq \Var\big(\proper_k(\bar{f})\bigmid
      \Param_k=\param\big)$.
\end{theorem} %}}}

%%%%%%%%%%%%%%%%%%%%
\begin{remark2} %{{{
In case of reversible chains, the condition in Theorem
\ref{thm:proper-clt-rev} \ref{item:proper-kv} is essentially
optimal, and the CLT relies on a result due to Kipnis and Varadhan
\cite{kipnis-varadhan}. The condition always holds when $(\Param_k)_{k\ge
1}$ is geometrically ergodic, for instance $(\Param_k)_{k\ge 1}$ is a
random-walk Metropolis algorithm and $\approxdens$ is light-tailed
\cite{jarner-hansen,roberts-tweedie}. In case $(\Param_k)_{k\ge 1}$ is
sub-geometric, such as polynomial, extra conditions are required; see
for instance \cite{jarner-roberts-heavytailed}. The condition
\ref{item:proper-mw} applies for non-reversible chains, and
relies on a result due to Maxwell and Woodroofe
\cite{maxwell-woodroofe}. If there exists $g\in L^2(\pi_a)$ 
which solves the Poisson equation $g-Pg = \meanfunc_{\bar{f}}$,
then \ref{item:proper-mw} holds,
but this is not necessary.
See also the review on Markov chain CLTs
by Jones \cite{jones}.
\end{remark2}
%}}}

When Theorem \ref{thm:proper-clt-rev} holds, we recall how a 
consistent confidence interval may be constructed.

%%%%%%%%%%%%%%%%%%%%
\begin{corollary} 
    \label{cor:is-ci} %{{{
Suppose that $\hat{a}_n$ is an estimator of the
integrated autocovariance of the sequence $(\xi_1(f),\ldots,\xi_n(f))$. 
If the conditions of Theorem \ref{thm:proper-clt-rev} hold
and $\hat{a}_n$ is consistent (see below), 
then for any $z_q>0$, 
\[
    \P\Big(\pi(f) \in 
    \Big[\E_n(f)  \pm  z_q
    \Big(\frac{\hat{s}_n}{n}\Big)^{1/2} \Big]\Big)
    \xrightarrow{n\to\infty}
    1 - 2 \Phi(z_q),
    \qquad\text{where}\qquad 
    \hat{s}_n \defeq \frac{\hat{a}_n}{\big(n^{-1} \sum_{k=1}^n
      \xi_k(1)\big)^2} 
\]
where $\Phi$ is the standard Gaussian distribution function.
\end{corollary} %}}}

The proof of Corollary \ref{cor:is-ci} is given in Appendix \ref{app:clt}.

By consistency of $\hat{a}_n$ we mean that it converges
to $\gamma_0 + 2\sum_{k\ge 1} \gamma_k$ in probability, which is
assumed to exist and be finite, where $(\gamma_k)_{k\ge 1}$ is the
stationary lag-$k$ autocovariance of the sequence $(\xi_k(f))_{k\ge
1}$. We refer the reader to consult, for instance, \citep{flegal-jones}
for details about consistent integrated autocovariance estimators.
  
Note that the latter term $\approxdens(v)$ in \eqref{eq:is-asvar}
contains the contribution of the `noise' in the IS estimates.
If the estimators $\proper_k(f)$ are made
increasingly accurate, in the sense that $\approxdens(v)$ becomes
negligible, the limiting case corresponds to an IS
corrected approximate MCMC and calculating averages over conditional
expectations $\meanfunc_{\bar{f}}(\param)$.
We conclude by relating the asymptotic variance with a straightforward
estimator.
%%%%%%%%%%%%%%%%%%%%
\begin{theorem}
    \label{thm:importance-var} %{{{
Suppose $f\in\Ls\cap L^2(\jointdens)$ and
$\approxdens(v)<\infty$
where $v$ is defined in Theorem \ref{thm:proper-clt-rev},
and also $\approxdens(m_{\unitfun}^{(2)})<\infty$.
Then, the estimator
\[
    \hat{v}_n \defeq \frac{\sum_{k=1}^n \big(\proper_k(f) -
    \proper_k(\unitfun)E_n(f)\big)^2 }{
    \big(\sum_{j=1}^n \proper_j(\unitfun) \big)^2}
\]
satisfies $n \hat{v}_n \to \approxdens(v+\meanfunc_{\bar{f}}^2)/c_w^2$ almost surely as $n\to\infty$.
\end{theorem} %}}}
Proof of Theorem \ref{thm:importance-var} is given in Appendix
\ref{app:clt}.

%%%%%%%%%%%%%%%%%%%%
\begin{remark2}
    \label{rem:variance-decomposition} %{{{
When $P$ corresponds to i.i.d.~sampling from $\pi_a$, 
the estimator $n\hat{v}_n$ in Theorem 
\ref{thm:importance-var} provides a
consistent estimate for the CLT variance $\sigma_f^2$. 
In most practical cases, $\Var(\meanfunc_{\bar{f}}, P)\ge
\pi_a(\meanfunc_{\bar{f}}^2)$ (which is always true when $P$ is
a positive operator \citep[cf.][]{andrieu-vihola-pseudo}), and then $n
\hat{v}_n$ provides an empirical lower bound of the asymptotic variance. 
Furthermore, it may be useful to inspect the `decomposition' of the
asymptotic variance into $n\hat{v}_n$ and the residual $\hat{s}_n
- n\hat{v}_n$. The former may be regarded as the `independent IS variance' 
and the residual may be interpreted as `excess marginal MCMC variance,' 
as it converges to
$2 c_w^{-2}\pi_a(\meanfunc_{\bar{f}}^2) \sum_{k\ge 1} \rho_k$,
where $\rho_k$ are the stationary autocorrelations of $(\xi_k(f))_{k\ge 1}$.
If the residual term is small relative to
$\hat{s}_n$, this suggests that either $\meanfunc_{\bar{f}}^2 \ll
v$, and/or that $\rho_k$ are small.
These inspections could provide insight to
choosing the parameters of the underlying marginal MCMC and the proper
weighting.
\end{remark2}
%}}}

%}}}

%%%%%%%%%%%%%%%%%%%%%%%%%%%%%%%%%%%%%%%%%%%%%%%%%%%%%%%%%%%%%%%%%%%%%%%%
\section{Jump chain estimators}
\label{sec:block} %{{{

Many MCMC algorithms such as the Metropolis-Hastings
include an accept-reject mechanism, which results in
blocks of repeated values
$\Param_k = \ldots = \Param_{k+b}$. 
In the context of IS type correction, and when the computational cost
of each estimate
$\proper_k$ is high, it may be desirable to construct only one estimator
per each \emph{accepted} state.
To formalise such an algorithm we consider the `jump chain'
representation of the approximate marginal chain 
\cite[cf.][]{douc-robert,doucet-pitt-deligiannidis-kohn,deligiannidis-lee}.
%%%%%%%%%%%%%%%%%%%%
\begin{definition2}[Jump chain]
    \label{def:jump} %{{{
          Suppose that $(\Param_k)_{k\ge 1}$ is Harris ergodic with
          respect to $\pi_a$. 
          The corresponding jump chain $(\tilde{\Param}_k)_{k\ge 1}$ with 
          holding times $(N_k)_{k\ge 1}$ is defined 
          as follows:
          \[
              \tilde{\Param}_k\defeq \Param_{\bar{N}_{k-1}+1}
              \qquad\text{and}\qquad
              N_k \defeq \inf\big\{j\ge 1\given \Param_{\bar{N}_{k-1}+j+1}\neq
                \tilde{\Param}_{k}\big\},
          \]
          where $\bar{N}_k \defeq \sum_{j=1}^k N_j$,
          and with $\bar{N}_0\equiv 0$.
      \end{definition2} %}}}
%%%%%%%%%%%%%%%%%%%%
\begin{remark2} %{{{
          If $(\Param_k)_{k\ge 1}$ corresponds to a Metropolis-Hastings
              chain, with non-diagonal proposal distributions $q$ (that is,
              $q(\param,\{\param\})=0$ for every $\param\in\paramspace$), then  the jump chain
              $(\tilde{\Param}_k)$ consists of the accepted states, and $N_k-1$
              is the number of rejections occurred at state
              $(\tilde{\Param}_k)$.
      \end{remark2} %}}}

Hereafter, we denote by $\alpha(\param) \defeq \P(\Param_{k+1}\neq
\Param_{k}\mid
\Param_{k}=\param)$ the overall acceptance probability at $\param$.
We consider next the practically important `jump IS' estimator,
involving a proper weighting for each accepted state.

%%%%%%%%%%%%%%%%%%%%
\begin{assumption}
    \label{a:block-natural-general} %{{{
Suppose that Assumption \ref{a:mcmc-is} holds, and let
$(\tilde{\Param}_k,N_k)_{k\ge 1}$ denote the corresponding
jump chain (Definition \ref{def:jump}).
Let $(\proper_k)_{k\ge 1}$
be a $\Ls$-proper weighting scheme,
where the variables $(M_k,\properweight_k^{(1:M_k)},\Latent_k^{(1:M_k)})$
in the scheme are now
allowed to depend on both $\tilde{\Param}_k$ and $N_k$, and
the conditions \ref{item:proper-unbiased} and
\ref{item:proper-general}. in Definition \ref{def:proper}
are replaced with the following:
\begin{enumerate}[label=(\roman*)]
    \item
$\E[\proper_k(f)\mid \Param_k=\param,N_k=n] = \meanfunc_f(\param)$
for all $n\in \N$ and $\approxdens(\meanfunc_f) = c_w \pi(f)$, 
\label{item:proper-unbiased-block}
\item $\approxdens(\bar{m}^{(1)})<\infty$ where
            $\bar{m}^{(1)}(\param)
              \defeq \sup_{n\in\N} \E\big[ |\proper_k(f)| \bigmid
                  \Param_k=\param, N_k=n\big]$.
  \label{item:proper-block}
\end{enumerate}
\end{assumption} %}}}

%%%%%%%%%%%%%%%%%%%%
\begin{theorem}
    \label{thm:block-natural-general} %{{{
Suppose Assumption \ref{a:block-natural-general} holds,
then,
\begin{equation}
    E_n(f) \defeq \frac{
      \sum_{k=1}^n N_k \proper_k(f)
      }{
      \sum_{j=1}^n
     N_j \proper_j(\unitfun)}
    \xrightarrow{n\to\infty} \jointdens(f)\qquad\text{almost surely}.
    \label{eq:block-natural}
\end{equation}
\end{theorem} %}}}

The proof follows from 
Lemma \ref{lem:super-general}
because $(\tilde{\Param}_k)$ is Harris ergodic with invariant
probability $\tilde\pi_a(\param) \propto \approxdens(\param)\alpha(\param)$;
see Proposition \ref{prop:jump-properties} in Appendix \ref{app:jump}.
Furthermore, 
the holding times $N_k\ge 0$ are, conditional on $(\tilde{\Param}_k)$,
independent  geometric random variables with parameter
$\alpha(\tilde{\Param}_k)$ (Proposition \ref{prop:jump-properties}), and
therefore $\E[N_k\mid \tilde{\Param}_k=\param] =
1/\alpha(\param)$.

%%%%%%%%%%%%%%%%%%%%
\begin{remark2} 
    \label{rem:jump} %{{{
Regarding Assumption \ref{a:block-natural-general}:
\begin{enumerate}[label=(\roman*)]
\item Condition \ref{item:proper-block} in Assumption
  \ref{a:block-natural-general}
  is practically convenient, because $\xi_k$ are usually chosen
  either as independent of $N_k$, or increasingly accurate in $N_k$
  (often taking $M_k$ proportional to $N_k$);
  see the discussion below. However, \ref{item:proper-block} is not optimal:
  it is not hard to find examples where
  the estimator is
  strongly consistent, even though $\bar{m}^{(1)}(\param)=\infty$ for some
  $\param\in\paramspace$.
\item 
  \label{item:jump-batch} 
  In case each $\proper_k$ is constructed as a mean of independent
  $(\proper_{k,1},\ldots,\proper_{k,N_k})$ (cf.~Proposition
  \ref{prop:convex-proper}), the jump chain estimator coincides with
  the simple estimator discussed in Section \ref{sec:asvar-clt} (at
  jump times). However, the jump chain estimator offers more 
  flexibility, which may allow for variance reduction, for instance by 
  using a single $m N_k$ particle
  filter (cf.~Appendix \ref{sec:ssm}) instead of 
  an average of $N_k$ independent $m$-particle filters, or by 
  stratification or control variates. 
\item Even though we believe that the estimators of the form
\eqref{eq:block-natural}
are often appropriate, we note that in some cases 
Rao-Blackwellised lower-variance estimators of $1/\alpha(\tilde{\Param}_j)$
may be used instead of $N_k$, as suggested in \cite{douc-robert}.
\end{enumerate}
\end{remark2} %}}}

Let us finally consider a central limit theorem corresponding
to the estimator in Theorem
\ref{thm:block-natural-general}, 
whose proof is given in Appendix \ref{app:jump}.
%%%%%%%%%%%%%%%%%%%%
\begin{theorem}
    \label{thm:block-clt} %{{{
   Suppose Assumption \ref{a:block-natural-general}
   holds, $(\tilde{\Param}_k)_{k\ge 1}$ is aperiodic, $f\in\Ls\cap
   L^2(\jointdens)$,
   \begin{equation}
       \approxdens\big(\alpha \tilde{m}^{(2)}\big)<\infty,
       \qquad\text{where}\qquad
                  \tilde{m}^{(2)}(\param)\defeq
                  \E\big[N_k^2 |\proper_k(\bar{f})|^2\bigmid
                  \tilde{\Param}_k=\param\big],
                  \label{eq:clt-tight}
  \end{equation}
  and one of the following holds:
  \begin{enumerate}[label=(\roman*)]
      \item
        \label{item:block-clt-rev}
        $(\Param_k)_{k\ge 1}$ is reversible and $\Var(\meanfunc_{\bar{f}} ,
        P)<\infty$.
      \item
        \label{item:block-clt-poisson}
        There exists $g\in L^2(\approxdens)$ satisfying the Poisson
        equation $g- Pg =
        \meanfunc_{\bar{f}}$.
  \end{enumerate}
   Then, the estimator $E_n(f)$ in \eqref{eq:block-natural} satisfies
   \[
       \sqrt{n}\big[E_n(f)-\jointdens(f)\big] \xrightarrow{n\to\infty} N(0,\sigma^2)
       \quad\text{in distribution, where}\quad
       \sigma^2 =
       \frac{\approxdens(\alpha)}{c_w^2}
    \Big[ \Var\big(\meanfunc_{\bar{f}}, P\big) +
    \approxdens(\alpha \tilde{v})
    \Big],
   \]
   and where $\tilde{v}(\param) \defeq \E\big[ N_k^2
   \Var\big(\proper_k(\bar{f}) \bigmid
   \tilde{\Param}_k=\param, N_k\big)\bigmid \tilde{\Param}_k=\param\big]$.
      \end{theorem} %}}}
Let us briefly discuss the conditions and implications of
Theorem \ref{thm:block-clt} under certain specific cases.
When the acceptance probability
      is bounded from below, $\inf_{\param} \alpha(\param)>0$,
      using a proper weighting $\proper_k$ independent of
      $N_k$ is `safe', because
      \[
          \tilde{v}(\param)\le \tilde{m}^{(2)}(\param) \le
  \frac{2-\alpha(\param)}{\alpha^2(\param)} b(\param);\quad
      b(\param)\defeq \sup_{n\ge 1} \E\big[|\proper_k (\bar{f})|^2\bigmid
      \tilde{\Param}_k=\param,N_k=n\big],
  \]
  and so $\approxdens(b)<\infty$ guarantees 
  \eqref{eq:clt-tight}.
    For example, if $(\Param_k)_{k\ge 1}$ is $L^2$-geometrically ergodic,
    then the acceptance probability is (essentially) bounded away from
    zero \cite{roberts-tweedie}, and
    $g\defeq \sum_{k\ge 0} P^k \meanfunc_{\bar{f}} \in L^2(\approxdens)$
    satisfies $g - Pg = \meanfunc_{\bar{f}}$, so that
    \ref{item:block-clt-poisson} is satisfied.

  When $\proper_k$ corresponds to an average of
      i.i.d.~$\proper_{k,1},\ldots$, $\proper_{k,N_k}$
      (cf.~Proposition \ref{prop:convex-proper}) which do not depend
      on $N_k$,
      \[
          \Var\big(\proper_k(\bar{f})  \bigmid
          \tilde{\Param}_k=\param,N_k\big)
          =\hat{v}(\param)/N_k;\quad 
          \hat{v}(\param) \defeq \Var\big(\proper_{k,1}(\bar{f})\bigmid
      \tilde{\Param}_k=\param).
      \]
      Then,
      $\approxdens(\alpha \tilde{v}) = \approxdens(\hat{v})$,
      which leads to an asymptotic variance that coincides with
      simple IS correction (cf.~Theorem \ref{thm:proper-clt-rev}).

%%%%%%%%%%%%%%%%%%%%
\begin{remark2} %{{{
Our condition in Theorem \ref{thm:block-clt}
\ref{item:block-clt-poisson}, requires the solution to the Poisson
equation, and is therefore more stringent than the Maxwell-Woodroofe
condition in Theorem \ref{thm:proper-clt-rev} \ref{item:proper-mw}. We
believe that the result holds more generally, but this would require
knowing properties of the jump chain $(\tilde{\Param}_k)_{k\ge 1}$,
which may be unavailable. 
Theorem \ref{thm:proper-clt-rev} \ref{item:proper-mw} may be applied
instead, if such properties are available.
\end{remark2} %}}}

%}}}

%%%%%%%%%%%%%%%%%%%%%%%%%%%%%%%%%%%%%%%%%%%%%%%%%%%%%%%%%%%%%%%%%%%%%%%%
\section{Pseudo-marginal approximate chain}
\label{sec:pseudo} %{{{

We next discuss how our limiting results still apply, in case the
approximate chain is a pseudo-marginal MCMC, as discussed in Section
\ref{sec:algorithmic-pseudo}. Let us first formalise 
a pseudo-marginal Markov chain $(\boldsymbol{\Theta}^\circ_k)_{k\ge 1}$, 
where $\boldsymbol{\Theta}^\circ_k = (\Param_k,\Phi_k)$ 
takes values on $\paramspace\times\Sp_\Phi$.
    Let $\boldsymbol{\Theta}^\circ_0 \in \paramspace\times \Sp_\Phi$ such that
    $U(\Phi_0)>0$, and
    for $k\ge 1$, iterate
    \begin{enumerate}
    \item[PM1] Generate $\tilde{\Param}_k \sim q(\Param_{k-1}, \uarg)$
      and $\tilde{\Phi}_k \sim Q_a(\tilde{\Param}_k, \uarg)$.
    \item[PM2] With probability
      $
          \min\big\{1,\frac{U(\tilde{\Phi}_k)
              q(\tilde{\Param}_k,\Param_{k-1})}{U(\Phi_{k-1})
              q(\Param_{k-1},\tilde{\Param}_k)}
            \big\},
      $
      accept and set
      $(\Param_k,\Phi_k)=(\tilde{\Param}_k,\tilde{\Phi}_k)$;
      otherwise reject and set $(\Param_k,\Phi_k) =
      (\Param_{k-1},\Phi_{k-1})$.
    \end{enumerate}
Above, $Q_a(\param,\uarg)$ defines a
(regular conditional)
distribution on (a measurable space) $\Sp_\Phi$, and
$U:\Sp_\Phi\to\R_+$ is a (measurable) function.
Under the following condition,
the pseudo-marginal Markov chain $(\boldsymbol{\Theta}^\circ_k)_{k\ge 1}$
is reversible with respect to the probability measure
$
    \approxdens^{\circ}(\ud \param, \ud \phi)
    \defeq \ud \param Q_a(\param, \ud \phi)U(\phi)/c_a,
$
which admits the marginal $\approxdens(\param)$ \cite[e.g.][]{andrieu-vihola-pseudo}:
%%%%%%%%%%%%%%%%%%%%
\begin{assumption}
    \label{a:approx-pseudo} %{{{
There exists a constant $c_a>0$ such that
for each $\param$, the random variable
$\Phi_\param\sim Q_a(\param, \uarg)$ satisfies
$\E[U(\Phi_\param)] = c_a \approxdens(\param)$.
\end{assumption} %}}}
\noindent In addition, $(\boldsymbol{\Theta}^\circ_k)_{k\ge 1}$ is easily
shown to be Harris ergodic under minimal conditions.

Hereafter, when we refer to the results in Sections
\ref{sec:simple}--\ref{sec:block}, we take the chain $(\boldsymbol{\Theta}^\circ_k)_{k\ge
  1}$ as the approximate chain (in place of $(\Theta_k)_{k\ge 1}$), 
with approximate marginal distribution 
$\approxdens^{\circ}(\ud \param, \ud \phi)$ (in place of
$\pi_a(\theta)\ud \theta$).
The following abstract minimal condition ensures
consistency of an IS type estimator. We discuss practically relevant
sufficient conditions later in Proposition \ref{prop:sufficient-pseudo-is}.

%%%%%%%%%%%%%%%%%%%%
\begin{assumption}
    \label{a:pseudo-proper} %{{{
Suppose Assumption \ref{a:mcmc-is} holds, 
$(\boldsymbol{\Theta}^\circ_k)_{k\ge 1}$ is Harris ergodic, 
$c_m>0$ is a constant,
and let $(P_k)_{k\ge 1}$
be conditionally independent given $(\boldsymbol{\Theta}^\circ_k)_{k\ge 1}$, such that the
distribution of each $P_k= (M_k,\properweightalt_k^{(1:M_k)},\Latent_k^{(1:M_k)})$
depends only on  $\boldsymbol{\Theta}^\circ_k$, 
where $M_k\in\N$,
$\properweightalt_k^{(i)}\in\R$ and $\Latent_k^{(i)}\in\latentspace$.
Define for any $f\in L^1(\jointdens)$,
$\properalt_k(f) \defeq \sum_{i=1}^{M_k} \properweightalt_k^{(i)}
    f(\Param_k,\Latent_k^{(i)})$,
and let
$\Ls\subset L^1(\jointdens)$ stand for all the functions for which
\begin{enumerate}[label=(\roman*)]
    \item
$
    \textstyle\iint Q_a(\theta, \ud \phi) \charfun{U(\phi)>0}
    \E[\properalt_k(f)\mid
\Param_k=\param, \Phi_k=\phi]
    \ud
\theta = c_m \jointdens(f),
$
and
\label{item:proper-pseudo-unbiased}
\item
$
    \textstyle\iint Q_a(\theta, \ud \phi) \charfun{U(\phi)>0}
    \E\big[ |\properalt_k(f)| \bigmid
                  \Param_k=\param,\Phi_k=\phi\big]
    \ud
\theta <\infty.
$
  \label{item:proper-pseudo-general}
\end{enumerate}
\end{assumption}
%}}}
%%%%%%%%%%%%%%%%%%%%
\begin{proposition}
    \label{prop:pseudo-approx-is} %{{{
Suppose Assumption \ref{a:approx-pseudo} and \ref{a:pseudo-proper} hold, and
$\unitfun\in \Ls$.
Then, Theorem \ref{thm:proper-consistency} holds with
\[
    \proper_k(f) \defeq \sum_{i=1}^{M_k} \properweight_k^{(i)}
    f(\Param_k,\Latent_k^{(i)}) \qquad\text{where}\qquad
    \properweight_k^{(i)} = \frac{\properweightalt_k^{(i)}}{U(\Phi_k)}.
\]
\end{proposition}
%}}}
The proof of Proposition \ref{prop:pseudo-approx-is}
follows by noting a proper weighting scheme involving the 
augmented approximate marginal distribution 
$\approxdens^\circ$ and target distribution $\jointdens^\circ$
(Lemma
\ref{lem:pseudo-approx-is}), and
Theorem \ref{thm:proper-consistency}.
%%%%%%%%%%%%%%%%%%%%
\begin{lemma}
    \label{lem:pseudo-approx-is} %{{{
Suppose the conditions of Proposition \ref{prop:pseudo-approx-is}
hold. Then, $\proper_k$ form a $\Ls^\circ$-proper weighting scheme,
with $\Ls^\circ \defeq \{f^\circ(\param,\phi,\latent) =
f(\param,\latent)\given f\in \Ls\}$, in the sense of Proposition
\ref{def:proper}, corresponding to
\begin{enumerate}[label=(\roman*)]
    \item approximate marginal $\approxdens^\circ(\ud \param, \ud \phi)
      = \ud \param Q_a(\param, \ud \phi)U(\phi)/c_a$,
    \item target $\jointdens^{\circ}\big((\ud\param, \ud \phi), \ud
      \latent\big)$ which admits
      the marginal $\jointdens(\param, \latent)\ud\param \ud\latent$.
\end{enumerate}
\end{lemma} %}}}
\begin{proof} %{{{
For any $f^\circ \in L^\circ$ and $\phi\in\Sp_\Phi$, let
$\nu_f(\param,\phi) \defeq \E[\properalt_k(f)\mid
\Param_k=\param, \Phi_k=\phi]$. Whenever $U(\phi)>0$, define
\[
    \meanfunc_{f^\circ}^\circ(\param,\phi) \defeq
    \E[\proper_k(f^\circ)\mid \Param_k=\param,\Phi_k=\phi]
    = \nu_f(\param,\phi)/U(\phi),
\]
and $\meanfunc_{f^\circ}^\circ(\param,\phi) \defeq 0$ otherwise.
We have
\[
    \approxdens^\circ(\meanfunc_{f^\circ}^\circ) = c_a^{-1}
    \textstyle \iint Q_a(\param,\ud \phi) \charfun{U(\phi)>0}
    \nu_f(\param,\phi) \ud \param
    = c_w \jointdens(f),
\]
by Assumption \ref{a:pseudo-proper}
\ref{item:proper-pseudo-unbiased}, where $c_w = c_m/c_a$.
We also have
\[
    m_{f^\circ}^{\circ(1)}(\param,\phi) \defeq
    \E[\proper_k(f^\circ)\mid \Param_k=\param,\phi_k=\phi]
    = |\nu_f(\param,\phi)|/U(\phi),
\]
so $\approxdens^\circ(m_{f^\circ}^{\circ(1)}) <\infty$ by
Assumption
\ref{a:pseudo-proper} \ref{item:proper-pseudo-general}.
\end{proof} %}}}

Let us finally consider different conditions, which
guarantee Assumption \ref{a:pseudo-proper}
\ref{item:proper-pseudo-unbiased}; the integrability Assumption \ref{a:pseudo-proper}
\ref{item:proper-pseudo-general} may be shown similarly.
%%%%%%%%%%%%%%%%%%%%
\begin{proposition}
    \label{prop:sufficient-pseudo-is} %{{{
Assumption \ref{a:pseudo-proper}
\ref{item:proper-pseudo-unbiased}
holds if
one of the following hold:
\begin{enumerate}[label=(\roman*)]
\item
  \label{item:pseudo-is-positive}
  For $\approxdens$-a.e.~$\param\in\paramspace$,
  $U(\Phi_\param)>0$ a.s.~and
  \begin{equation}
      \E[\properalt_k(f)\mid
\Param_k=\param]
= c_m \margdens(\param) f^*(\param),
\label{eq:unbiasedness-general-pseudo}
  \end{equation}
  where 
  $\E[\properalt_k(f)\mid
\Param_k=\param]
     = \textstyle \int Q_a(\param,\ud\phi)
      \E[\properalt_k(f)\mid
\Param_k=\param, \Phi_k=\phi]$.
\item
      \label{item:pseudo-is-independent}
$\properalt_k$ only depend on $\Param_k$, and
 for $\approxdens$-a.e.~$\param\in\paramspace$,
\[
    \E[\properalt_k(f)\mid
\Param_k=\param]
     = c_m \margdens(\param) f^*(\param)/p(\param),
\]
where $p(\param) \defeq \P(U(\Phi_\param)>0)$ with $\Phi_\param \sim
Q_a(\param,\uarg)$.
\item \label{item:pseudo-is-support}
  For $\approxdens$-a.e.~$\param\in\paramspace$
  \eqref{eq:unbiasedness-general-pseudo} holds, and
  $U(\phi)=0$ implies $\E[\properalt_k(f)\mid
\Param_k=\param, \Phi_k=\phi]=0$.
\end{enumerate}
\end{proposition} %}}}
\begin{proof} %{{{
Note that \ref{item:pseudo-is-positive} implies
\ref{item:pseudo-is-support}, under which
\[
    \textstyle
    \iint Q_a(\param,\ud\phi) \charfun{U(\phi)>0}
    \nu_f(\param,\phi) \ud\param
    = c_m \int
    \margdens(\param) f^*(\param) \ud \param
    = c_m \pi(f),
\]
where $\nu_f(\param,\phi) = \E[\properalt_k(f)\mid
\Param_k=\param,\Phi_k=\phi]$.

In case of \ref{item:pseudo-is-independent}, we have
$\nu_f(\param,\phi) = \E[\properalt_k(f)\mid \Param_k=\param]$ and so
\[
   \textstyle \int Q_a(\param,\ud\phi) \charfun{U(\phi)>0}
    \nu_f(\param,\phi)
    = c_m \margdens(\param) f^*(\param).%\qedhere
\]
\end{proof} %}}}

%%%%%%%%%%%%%%%%%%%%
\begin{remark2}
    \label{rem:positivity-essential} %{{{
Proposition \ref{prop:sufficient-pseudo-is}
\ref{item:pseudo-is-positive}
is the most straightforward in the latent variable context,
and often sufficient, since we may choose a positive $U(\phi)$
(e.g.~by considering inflated $\tilde{U}(\phi) = U(\phi)+\epsilon$
instead).
Proposition \ref{prop:sufficient-pseudo-is}
\ref{item:pseudo-is-independent}
  may be used directly to verify the validity of an
  MCMC version of the lazy ABC algorithm \cite{prangle}.
  It also
  demonstrates why positivity plays a
key role: if only \eqref{eq:unbiasedness-general-pseudo} is assumed
and $p(\param)$ is non-constant, then $p(\param)$ must be accounted
for, or else we end up with biased estimators targeting a marginal
proportional to $\margdens(\param)p(\param)$.
Proposition \ref{prop:sufficient-pseudo-is} \ref{item:pseudo-is-support}
demonstrates that strict positivity is
not necessary, but in this case a delicate dependency structure is
required.
\end{remark2}
%}}}

%}}}

%%%%%%%%%%%%%%%%%%%%%%%%%%%%%%%%%%%%%%%%
\section{State space models and linear-Gaussian state dynamics}
\label{sec:lin-gauss-ssms} %{{{

State space models (SSM) are latent variable models
which are commonly used in time series analysis
\cite[cf.][]{cappe-moulines-ryden}.
In the setting of Section \ref{sec:latent},
SSMs are parametrised
by $\param\in\paramspace$, and $\latent = \ssmhid_{1:T} \in
\latentspace = \Sp_\ssmhid^T$ and
$\obs = \ssmobs_{1:T}\in \obsspace = \Sp_\ssmobs^T$, and
\[
  \mu^{(\param)}(\latent)
  = \prod_{t=1}^T
  \mu_t^{(\param)}(\ssmhid_t\mid \ssmhid_{t-1})
  \qquad\text{and}\qquad
  g^{(\param)}(\obs\mid\latent) =
  \prod_{t=1}^T g_t^{(\param)}(\ssmobs_t\mid\ssmhid_t),
\]
where, by convention, $\mu_1^{(\param)}(\ssmhid_1\mid \ssmhid_0) \defeq
\mu_1^{(\param)}(\ssmhid_1)$. That is, the latent states
$\Ssmhid_{1:T}$ form a Markov
chain with initial density $\mu_1^{(\param)}$ and state transition
densities $\mu_t^{(\param)}$, and
$g_t^{(\param)}$ define the model for the observations $\Ssmobs_{1:T}$
given $\Ssmhid_{1:T}$.

Appendix \ref{sec:ssm} reviews general techniques to construct 
$\properweightalt^{(1:m)}_\param$ and $\Latent^{(1:m)}_\param$
for which $\properalt_\param(h) \defeq \sum_{i=1}^m
    \properweightalt^{(i)}_\param h(\Latent^{(i)})$ satisfy:
\begin{align}
    \E[\properalt_\param(h)]
    &= \int p^{(\param)}(\ssmhid_{1:T}, \ssmobs_{1:T})
    h(\ssmhid_{1:T}) \ud \ssmhid_{1:T},
    \label{eq:ssm-proper}
\end{align}
for any $\param$ and for some class of functions $h:\Sp_\ssmhid^T\to\R$.
These random variables lead directly to a proper
weighting; see Corollary \ref{cor:proper-ssm} in Appendix \ref{sec:ssm}.

We focus next on a special case of the general SSM, 
where both $\Sp_\ssmhid$ and $\Sp_\ssmobs$ are
Euclidean and $\mu_t^{(\param)}$ are linear-Gaussian, but
the observation models
$g_t^{(\param)}$ may be non-linear and/or non-Gaussian,
taking the form
\begin{equation*}
    g_t^{(\param)}(\ssmobs_t \mid  \ssmhid_t) =
\eta_t^{(\param)}(\ssmobs_t\mid H_t^{(\param)} \ssmhid_t).
\end{equation*}
Our setting covers exponential family observation models
with Gaussian, Poisson, binomial,
negative binomial, and Gamma distributions, and a stochastic
volatility model.
This class contains a large number of commonly used
models, such as structural time series
models, cubic splines, generalised linear mixed models, and classical
autoregressive integrated moving average models.

%%%%%%%%%%%%%%%%%%%%%%%%%%%%%%%%%%%%%%%%
\subsection{Marginal approximation}
\label{sec:lin-gauss-marginal} %{{{

The scheme we consider here is based on 
\cite{shephard-pitt,durbin-koopman1997},
and relies on a Laplace approximation
$p^{(\param)}_a(\ssmhid_{1:T}, \tilde{\ssmobs}_{1:T}^{(\param)})
= \mu^{(\param)}(\ssmhid_{1:T})
\tilde{g}^{(\param)}(\tilde{\ssmobs}_{1:T}^{(\param)}\mid
\ssmhid_{1:T})$, where
$\tilde{g}^{(\param)}(\tilde{\ssmobs}_{1:T}^{(\param)}\mid \ssmhid_{1:T})
\defeq \prod_{t=1}^T
\tilde{g}_t^{(\param)}(\tilde{\ssmobs}_t^{(\param)}\mid \ssmhid_t)$.
The linear-Gaussian terms $\tilde{g}_t$ approximate $g_t$
in terms of pseudo-observations $\tilde{\ssmobs}_{t}^{(\param)}$ and
pseudo-covariances $R_t^{(\param)}$, which are
found by an iterative process, which we detail next for a fixed $\param$.
Denote $D_t^{(n)}(\ssmhid_t) \defeq \frac{\partial^n}{\partial^n z_t}
  \log \eta_t^{(\param)}(\ssmobs_t \mid
\ssmhid_t)$, and assume that $\tilde \ssmhid_{1:T}$ is an initial estimate
for the mode $\hat{\ssmhid}_{1:T}^{(\param)}$ of
$p^{(\param)}(\ssmhid_{1:T} \mid \ssmobs_{1:T})$ following:
\begin{enumerate}
\item $R_t^{(\param)} = - [D_t^{(2)}(H_t^{(\param)}\tilde{\ssmhid}_t)]^{-1}$
and 
$\tilde \ssmobs_t^{(\param)} = H_t^{(\param)} \tilde \ssmhid_t + R_t^{(\param)}
D_t^{(1)}(H_t^{(\param)} \tilde{\ssmhid}_t)$
\item Run the Kalman filter and smoother for the model with 
  $g_t^{(\param)}(\ssmobs_t\mid \ssmhid_t)$ replaced by
$    \tilde{g}_t^{(\param)}(\tilde{\ssmobs}_t^{(\param)}\mid \ssmhid_t)
    = N(\tilde{\ssmobs}_t^{(\param)}; H_t^{(\param)}\ssmhid_t, R_t^{(\param)})$
and set $\tilde{\ssmhid}_{1:T}$ to the smoothed mean.
\end{enumerate}
These steps are then repeated until convergence,
which is typically quick: often less than 10 iterations are enough
\cite{durbin-koopman2000}.

Consider the following decomposition of the marginal likelihood:
\begin{equation}
    \label{logp}
        L(\param)
        = \tilde{L}_a(\param)
        \frac{g^{(\param)}(\ssmobs_{1:T} \mid \hat
          \ssmhid_{1:T}^{(\param)})}{
              \tilde g^{(\param)}(\tilde{\ssmobs}_{1:T}^{(\param)}\mid
              \hat \ssmhid_{1:T}^{(\param)})}
         \E\left[\frac{g^{(\param)}(\ssmobs_{1:T} \mid
           \Ssmhid_{1:T})
           / g^{(\param)}(\ssmobs_{1:T} \mid \hat
           \ssmhid_{1:T}^{(\param)})
           }{\tilde g^{(\param)}(\tilde{\ssmobs}_{1:T}^{(\param)}\mid \Ssmhid_{1:T})
           / \tilde g^{(\param)}(\tilde{\ssmobs}_{1:T}^{(\param)} \mid \hat
           \ssmhid_{1:T}^{(\param)})
           }\right],
\end{equation}
where
$\tilde{L}_a(\param) \defeq \int p^{(\param)}_a(\ssmhid_{1:T}, \tilde
\ssmobs_{1:T}^{(\param)})
\ud \ssmhid_{1:T}$ is the marginal likelihood (from the Kalman
filter), and the expectation is taken with respect to the approximate
smoothing distribution
$p_a^{(\param)}(\ssmhid_{1:T}\mid \tilde{\ssmobs}_{1:T}^{(\param)}) =
p^{(\param)}_a(\ssmhid_{1:T}, \tilde\ssmobs_{1:T}^{(\param)})/\tilde{L}_a(\param)$.
If the pseudo-likelihoods $\tilde{g}_t^{(\param)}$ are nearly proportional
to the true likelihoods
$g_t^{(\param)}$ around the mode of
$p_a^{(\param)}(\ssmhid_{1:T}\mid \ssmobs_{1:T})$,
the expectation in \eqref{logp} is close to one. Our 
approximation is based on dropping
the expectation in \eqref{logp}:
$    L_a(\param) \defeq \tilde{L}_a(\param)
    g^{(\param)}(\ssmobs_{1:T} \mid \hat \ssmhid_{1:T})/\tilde
      g^{(\param)}(\tilde \ssmobs_{1:T}^{(\param)}\mid \hat
      \ssmhid_{1:T}).$
The same approximate likelihood $L_a(\param)$ was also used in a
maximum likelihood setting by \cite{durbin-koopman2012} as an initial
objective function before more expensive importance sampling based
maximisation was done.

The evaluation of the approximation $L_a(\theta)$ above requires a
reconstruction of the Laplace approximation for each value of
$\theta$. We call this \emph{local approximation}, and consider also a
faster \emph{global approximation} variant, where the
pseudo-observations and covariances are constructed only once, at the
maximum likelihood estimate of $\param$.

%}}}

%%%%%%%%%%%%%%%%%%%%%%%%%%%%%%%%%%%%%%%%
\subsection{Proper weighting schemes}
\label{sec:lin-gauss-proper} %{{{

The simplest approach to construct a proper weighting scheme based on
the Laplace approximations is to use the approximate smoothing
distribution $p_a^{(\param)}(\ssmhid_{1:T}\mid \ssmobs_{1:T})$ as IS
proposal. We consider such scheme using the simulation
smoother \cite{durbin-koopman2002} with one antithetic variable, which 
we call \textbf{SPDK}, following \cite{shephard-pitt}.

We consider also several variants of $M_t$ and $G_t$ in the particle
filter discussed in Appendix \ref{sec:ssm}. The bootstrap filter
\cite{gordon-salmond-smith}, abbreviated as
\textbf{BSF}, uses $M_t=\mu_t^{(\param)}$ and
$G_t=g_t^{(\param)}(\ssmobs_t\mid \uarg)$, and hence does not
rely on an approximation. Inspired by the developments in
\cite{whiteley-lee,guarniero-johansen-lee}, we consider also the choice
\begin{equation*}
M_t(\ssmhid_t\mid \ssmhid_{1:t-1}) = p_a^{(\param)}(\ssmhid_t\mid
\ssmhid_{t-1},\ssmobs_{1:T}) ,\qquad\text{and}\qquad G_t(\ssmhid_{1:t})
= g_t^{(\param)}(\ssmobs_t\mid \ssmhid_t) /
\tilde{g}_t^{(\param)}(\tilde{\ssmobs}_t\mid \ssmhid_t),
\end{equation*}
where $p_a^{(\param)}(\ssmhid_t\mid
\ssmhid_{t-1},\ssmobs_{1:T}) = p_a^{(\param)}(\ssmhid_t\mid
\ssmhid_{1:t-1},\ssmobs_{1:T})$ are conditionals of
$p_a^{(\param)}(\ssmhid_{1:T}\mid \ssmobs_{1:T})$. This would be
optimal in our setting if the $G_t$ were constants
\cite{guarniero-johansen-lee}. As they are often approximately so, we
believe that this choice, which we call
\textbf{$\psi$-APF} following 
\cite{guarniero-johansen-lee}, can provide substantial benefits over BSF.

%}}}

%}}}

%%%%%%%%%%%%%%%%%%%%%%%%%%%%%%%%%%%%%%%%
\section{Discretely observed diffusions}
\label{sec:diffusion-ssms} %{{{

In many applications, for instance in finance or physical systems
modelling, the SSM state transitions arise naturally from
a continuous time diffusion model, such as
\begin{equation*}
    \ud \tilde{\Ssmhid}_t =
    m^{(\theta)}(t,\tilde{\Ssmhid}_t) \ud t +
    \sigma^{(\theta)}(t,\tilde{\Ssmhid}_t) \ud B_t,
\end{equation*}
where $B_t$ is a (vector valued) Brownian motion and where
$m^{(\theta)}$ and $\sigma^{(\theta)}$ are functions (vector and
matrix valued, respectively). The latent variables $\Latent =
(\Ssmhid_1,\ldots,\Ssmhid_T)$ are assumed to follow the law of
$(\tilde{\Ssmhid}_{t_1},\ldots,\tilde{\Ssmhid}_{t_T})$, so
$\mu_k^{(\param)}$ would ideally be the transition densities of
$\tilde{\Ssmhid}_{t_k}$ given $\tilde{\Ssmhid}_{t_{k-1}}$. These
transition densities are generally unavailable (for non-linear
diffusions), but standard time-discretisation schemes allow for
straightforward approximate simulation \cite[cf.][]{kloeden-platen}.
The denser the time-discretisation the less bias,
but the computational complexity of the simulation is
higher --- generally proportional to the size of the mesh.

The MCMC-IS may be applied to speed up the inference of discretely
observed diffusions by the following simple two-level approach.
The `true' state transition $\mu_t^{(\param)}$ are based on
`fine enough' discretisations, which are assumed to ensure a
negligible bias, but which are expensive to simulate. Cheaper `coarse'
discretisation corresponds to transitions $\hat{\mu}_t^{(\param)}$.

Because neither of the models admit exact calculations, we may only
use a pseudo-marginal approximate chain as discussed in Sections
\ref{sec:algorithmic-pseudo} and \ref{sec:pseudo}). More specifically,
we may use the bootstrap filter (Appendix \ref{sec:ssm}) with SSM
$(\hat{\mu}_t^{(\tilde{\Param}_k)}, g_t^{(\tilde{\Param}_k)})$ to
generate the likelihood estimators $\tilde{U}_k$ in Phase 1', and in Phase
2', we may use bootstrap filters for
SSM $(\mu_t^{(\Param_k)},g_t^{(\Param_k)})$ to generate
$(\properweightalt_k^{(i)},\Latent_k^{(i)})$.

Assuming that the observation model satisfies $g_t^{(\param)}>0$
guarantees the validity of this scheme, because then $\tilde{U}_k>0$
(see Proposition \ref{prop:sufficient-pseudo-is}
\ref{item:pseudo-is-positive}).
It is most straightforward to simulate the bootstrap filters in Phases
1' and 2'
independent of each other, but they may be made dependent as well, by
using a coupling strategy \cite[cf.][]{sen-thiery-jasra}.
The correction phase could be also based on
exact sampling for diffusions
\cite{beskos-papaspiliopoulos-roberts-fearnhead}, 
which allow for
elimination of the discretisation bias entirely.

The recent work \cite{franks-jasra-law-vihola} details how
unbiased inference is also possible with IS type correction, using
randomised multilevel Monte Carlo.

%}}}

%%%%%%%%%%%%%%%%%%%%%%%%%%%%%%%%%%%%%%%%%%%%%%%%%%%%%%%%%%%%%%%%%%%%%%%%
\section{Experiments}
\label{sec:exp} %{{{

%{{{

We did experiments for our generic framework with SSMs, using Laplace
approximations (Section \ref{sec:lin-gauss-ssms}) and an approximation
based on coarsely discretised diffusions (Section
\ref{sec:diffusion-ssms}).  We compared several approaches in our
experiments:
\begin{description}
\item[AI] Approximate inference with MCMC targeting  $\approxdens(\param)$,
  and for each accepted $\tilde{\Param}_k$, sampling one realisation from
  $\tilde p^{(\tilde{\Param}_k)}(\ssmhid_{1:T} \mid \ssmobs_{1:T})$.
\item[PM] Pseudo-marginal MCMC with $m$ samples targeting directly
  $\jointdens(\param,\latent)$.
\item[DA] Two-level delayed acceptance pseudo-marginal MCMC with first
  stage acceptance based on
  $\approxdens(\param)$ and with target $\jointdens(\param,\latent)$.
\item[IS1] Jump chain IS correction with $m N_k$ samples for each accepted
  $\tilde{\Param}_k$.
\item[IS2] Jump chain IS correction with $m$ samples for each accepted
  $\tilde{\Param}_k$.
\end{description}
The IS1 algorithm is similar to simple IS estimator 
\eqref{eq:corrected-estimator}, but is expected
to be generally safer; see Remark \ref{rem:jump} \ref{item:jump-batch}
Except for AI, all the algorithms are asymptotically exact.
Ignoring the effects of parallel implementation, the average
computational complexity, or cost, of DA and IS2 are roughly comparable, and
we have similar pairing between PM and IS1. However, as the weighting
in IS methods is based only on the post-burn-in chain, the IS methods
are generally somewhat faster.

We used a random walk Metropolis algorithm for $\approxdens$ with a
Gaussian proposal distribution, whose covariance was adapted during
burn-in following \cite{vihola-ram}, targeting the acceptance rate
0.234. In DA, the adaptation was based on the first stage acceptance probability,
which was also used in PM for SPDK and $\psi$-APF variants as this led to more robust adaptation.

All the experiments were conducted in \texttt{R} \cite{r-core} using
the \texttt{bssm} package which is available online
\cite{helske-vihola}. The experiments were run on a Linux server with
four 24-core Intel Xeon E7-8890 2.6GHz processors with total 2.1TB
of RAM. 

In each experiment, we calculated the Monte Carlo estimates several
times independently, and the inverse relative efficiency (IRE) was
reported. The IRE, defined as the mean square error (MSE) of the
estimate multiplied by the average computation time, provides a
justified way to compare Monte Carlo algorithms with different costs
\cite{glynn-whitt}. 

Further details and results of the experiments may be found in the
preprint version of our article \cite{vihola-helske-franks-preprint}.

%}}}

%%%%%%%%%%%%%%%%%%%%%%%%%%%%%%%%%%%%%%%%
\subsection{Laplace approximations} %{{{

%{{{

In case of Laplace approximations, the maximum likelihood estimate of
$\param$ was always used as the starting value of MCMC.
We used sub-sampling as in
Proposition \ref{prop:randomise-proper}, and sampled one
trajectory $\Ssmhid_{1:T}$ per each accepted state.
We tested the exact methods with three different IS
correction schemes,
SPDK, BSF and $\psi$-APF, described
in Section \ref{sec:lin-gauss-proper}.
For BSF and $\psi$-APF, the filter-smoother estimates as in Proposition
  \ref{prop:particle-proper} \ref{item:filter-smoother} were used.
When calculating the MSE, we used the average over all estimates from all
unbiased algorithms as the ground truth.

For all the exact methods, we chose the IS accuracy parameter $m$
based on a pilot experiment, following the guidelines 
for optimal tuning of pseudo-marginal MCMC in
\cite{doucet-pitt-deligiannidis-kohn}.
More specifically, $m$ was set so that the
standard deviation of the logarithm of the likelihood estimate,
denoted with $\delta$, was around 1.2 in the neighbourhood of the
posterior mean of $\param$.  
We kept the same $m$ for all
methods, for comparability, even though in some cases optimal choice
might differ \cite{sherlock-thiery-lee}. 

%}}}

%%%%%%%%%%%%%%%%%%%%%%%%%%%%%%
\subsubsection{Poisson observations}
\label{poisson} %{{{

Our first model is of the following form:
\begin{equation*}
              g_t^{(\param)}(\ssmobs_t \mid \ssmhid_t) =
              \mathrm{Poisson}(\ssmobs_t; e^{u_t}), \quad\text{and}\quad
              \begin{pmatrix}
                  u_{t+1} \\
                  v_{t+1}\end{pmatrix}
              = \begin{pmatrix}
              u_t + v_t + \sigma_{\eta}\eta_t \\
              v_t + \sigma_{\xi}\xi_t
              \end{pmatrix},
\end{equation*}
with $\Ssmhid_1=(U_1,V_1)\sim N(0, 0.1 I)$, where
$\xi_t,\eta_t\sim N(0,1)$. For testing our algorithms, we
simulated a single set of observations $\ssmobs_{1:100}$ from this model
with $\Ssmhid_1=(0, 0)$ and $\param = (\sigma_{\eta},
\sigma_{\xi})= (0.1, 0.01)$.
We used a uniform prior $U(0, 2s)$ for the
parameters, where the cut-off parameter $s$ was
set to $1.6$ based on the sample standard deviation of
$\log(\ssmobs_{1:T})$, where zeros were replaced with 0.1. Results were not
sensitive to this upper bound.

Based on a pilot optimisation, we set
$m=10$ for SPDK and $\psi$-APF, and  $m=200$ for BSF, leading to $\delta\approx 0.05$ for SPDK, $\delta\approx 0.1$ for $\psi$-APF, and $\delta\approx 1.2$ for BSF.
For all algorithms, we used 100,000 MCMC iterations with the first
half discarded as burn-in.
We ran all the algorithms independently
1000 times. 

Table \ref{table:poisson} shows the IREs, which are 
re-scaled such that all IREs of PM-BSF equal one.
The overall acceptance rate of DA-BSF was around 0.101, 
and 0.218-0.234 for all others.
All exact methods led to essentially the same overall mean
estimate $(0.093, 0.016, -0.075, 2.618)$ for
($\sigma_{\eta}$, $\sigma_{\xi}$, $u_1$, $u_{100}$), in contrast with
AI showing some bias on $(u_1, u_{100})$, with overall mean estimates 
$(-0.063, 2.629)$ and $(-0.064, 2.631)$ with local and
global approximation, respectively. 
IS2-BSF clearly outperformed DA-BSF in terms of IRE,
because of the burn-in benefit.  Similarly, IS1-BSF outperformed
PM-BSF by a clear margin.
With SPDK and $\psi$-APF, the IS1 and IS2 outperformed the PM and DA
alternatives, but with a smaller margin because of
smaller overall execution times.
There were no significant differences between the SEs of local and global
variants, but the global one was faster leading to smaller IREs. The execution times of SPDK were slightly less than $\psi$-APF due to the use of antithetic variable which increased the relative efficiency of SPDK algorithm.

Using Corollary \ref{cor:is-ci}, we constructed 95\% confidence
intervals for our IS-MCMC estimators, and computed the average
coverage of these intervals over all replications. The coverages were
between 0.93 and 0.99, averaging the nominal 0.95 over all methods,
with no clear differences between them. In addition, we computed the
`variance decomposition' as suggested in Remark
\ref{rem:variance-decomposition}. That is, we calculated the
estimator $n\hat{v}_n$ of Theorem \ref{thm:importance-var}, and
calculated the proportion of it with respect to the overall asymptotic
variance estimated with $\hat{s}_n$ of Corollary \ref{cor:is-ci},
using the integrated autocovariance estimator suggested by
\citet{sokal-notes}.  For the hyperparameters $\sigma_{\eta}$ and
$\sigma_{\xi}$, these proportions were around 0.5 and 0.6,
respectively, except for IS2-BSF which resulted proportions 0.7 and
0.8. This suggests that the MCMC autocorrelations and the IS
correction contribute roughly equally to the overall uncertainty
For the
latent states $(u_1, u_{100})$ the proportion was close to 1 in all
cases. This is expected, because in the case of the latent variables, 
the centred conditional expectation terms $\meanfunc_{\bar{f}}^2$ are
expected to be small relative to the conditional variance $v$.

%%%%%%%%%%%%%%%%%%%%
\begin{table}[t]
	\centering
  \caption{IREs for the Poisson model, with local (top) and global
    (bottom) approximations. Times are in seconds. For PM-BSF, IREs are
    one and time 811s.}
	\label{table:poisson} %{{{
        \small
      \begin{threeparttable}
      \begin{tabular}{p{1.8em}ccccccccccccc}
      \toprule
&& \multicolumn{3}{c}{BSF}
      & \multicolumn{4}{c}{SPDK} 
      & \multicolumn{4}{c}{$\psi$-APF} \\
%      \hiderowcolors
      \cmidrule(lr){3-5}
      \cmidrule(lr){6-9}
      \cmidrule(lr){10-13}
      & AI & DA & IS1 & IS2 
      & PM & DA & IS1 & IS2 
      & PM & DA & IS1 & IS2 \\
      	\midrule
      %  \multirow{5}{*}{\rotatebox[origin=c]{90}{Local}} &
%      \headrow
        Time
        & 82 & 273 & 540 & 178
        & 137 & 93 & 102 & 86 
        & 206  & 110 & 111 & 94 \\
%      	\cmidrule{2-14}
%        \midrule
%      \showrowcolors
        $\sigma_{\eta}$ 
        & 0.031 & 0.354 & 0.240 & 0.134
        & 0.051 & 0.035 & 0.037 & 0.033 
        & 0.077 & 0.044 & 0.043 & 0.034 
        \\
        $\sigma_{\xi}$ 
        & 0.037 & 0.388 & 0.290 & 0.159 
        & 0.063 & 0.043 & 0.049 & 0.042 
        & 0.091 & 0.054 & 0.056 & 0.042
        \\
         $u_1$ 
         & 0.964 & 0.795 & 0.544 & 0.414
         & 0.122 & 0.076 & 0.090 & 0.040 
         & 0.176 & 0.095 & 0.093 & 0.072 
         \\
         $u_{100}$ 
         & 1.713 & 0.825 & 0.523 & 0.382
         & 0.113 & 0.079 & 0.082 & 0.042 
         & 0.172 & 0.101 & 0.095 & 0.067
         \\
        \midrule
        %\multirow{5}{*}{\rotatebox[origin=c]{90}{Global}} &
%      \showrowcolors
%        \headrow
%      \hline
%      \hline
        Time 
    & 23 & 214 & 482 & 120
    & 78 & 35 & 44 & 28
    & 147 & 51 & 53 & 36 \\
%		\cmidrule{2-14}
%    \midrule
$\sigma_{\eta}$ 
 & 0.012 & 0.284 & 0.220 & 0.091 
 & 0.033 & 0.015 & 0.019 & 0.013 
 & 0.065 & 0.024 & 0.022 & 0.017
 \\
$\sigma_{\xi}$ 
 & 0.065 & 0.340 & 0.268 & 0.118 
 & 0.041 & 0.019 & 0.024 & 0.015 
 & 0.074 & 0.029 & 0.027 & 0.019 
 \\
$u_1$ 
 & 0.211 & 0.601 & 0.440 & 0.271 
 & 0.061 & 0.028 & 0.032 & 0.013 
 & 0.122 & 0.043 & 0.039 & 0.027 
 \\
$u_{100}$ 
 & 0.660 & 0.592 & 0.433 & 0.240 
 & 0.067 & 0.031 & 0.038 & 0.015 
 & 0.126 & 0.048 & 0.042 & 0.023 
 \\
      	\bottomrule
\end{tabular}
\end{threeparttable}
\end{table}
    %}}}

%}}} % {poisson}

%%%%%%%%%%%%%%%%%%%%%%%%%%%%%%
\subsubsection{Stochastic volatility model}
\label{sv-model} %{{{

Our second illustration is more challenging, involving analysis of real time
series: the daily log-returns for the S\&P
index from 4/1/1995 to 28/9/2016, with total number of observations $T =
5473$. The data was analysed using the following stochastic volatility
(SV) model:
\begin{equation*}
    \Ssmobs_t\mid\Ssmhid_t \sim N(0,e^{\Ssmhid_t}), \qquad
    \Ssmhid_{t+1}\mid \Ssmhid_t \sim N(\nu + \phi (\Ssmhid_t - \nu),
    \sigma_{\eta}^2),
\end{equation*}
with $\Ssmhid_1 \sim N(\nu, \sigma^2_{\eta} / (1-\phi^2))$.
We used a
uniform prior on $[-0.9999, 0.9999]$ for $\phi$, a
half-Gaussian prior with standard deviation 5 for $\sigma_{\eta}$,
and a zero-mean Gaussian prior with standard deviation 5 for $\nu$.
SPDK was expected to be problematic, due to its well-known exponential
deterioration in $T$, unlike the particle filter which often scales
much better in $T$ \cite{whiteley-stability}. In addition, it is
known that for this particular model, the importance weights may have
large variability
\cite{pitt-tran-scharth-kohn2013,koopman-shephard-creal2009}.  While
in principle $\psi$-APF may also be affected by such
fluctuations, we did not observe any problems with it in our
experiments.

Based on our pilot experiment, we chose $m=10$ for $\psi$-APF, 
$m=70$ for SPDK and $m=3400$ for BSF, 
which all led to $\delta\approx 1.1$.  We used 
100,000 MCMC iterations with the first half discarded as burn-in,
and 100 independent replications. 
the IREs re-scaled here with respect to DA-BSF are shown in
Table \ref{table:sv}. The PM and IS1 were not tested because of their
high costs. The results with global approximation are shown only
for AI, and indicate significant computational savings.
The parallelisation with 8
cores dropped the execution time nearly ideally.
The total acceptance rates were 0.1 for DA-BSF,  
PM-SPDK and DA-$\psi$-APF, 0.06 for DA-SPDK, and 0.15 for PM-$\psi$-APF.

\begin{table}[t]
	\centering \caption{IREs for SV model. Times are in hours.
    AI\textsuperscript{G} is with global approximation and
    IS2\textsuperscript{8} is with 8 parallel cores.
    For DA-BSF, IREs are
    one and time 37h. }
	\label{table:sv} %{{{
        \small
  \begin{threeparttable}
	\begin{tabular}{p{1.8em}cccccccccccc}
    \toprule
    &&& \multicolumn{2}{c}{BSF}
    & \multicolumn{4}{c}{SPDK}
    & \multicolumn{4}{c}{$\psi$-APF} \\
    \cmidrule(lr){4-5}
    \cmidrule(lr){6-9}
    \cmidrule(lr){10-13}
 %   \headrow
     & AI & AI\textsuperscript{G} & IS2 & IS2\textsuperscript{8}
     & PM & DA & IS1 & IS2
     & PM & DA & IS1 & IS2 \\
		\midrule
%     \headrow
		 Time 
     & 0.9 & 0.2 & 19.6 & 3.1
     & 3.8 & 1.6 & 2.4 & 1.2
     & 2.1 & 1.2 & 1.2 & 1\\
%                          \midrule
 $\phi$ 
 & 0.114 & 0.114 & 0.328 & 0.068    
 & 0.634 & 2.901 & 0.346 & 0.645 
 & 0.035 & 0.036 & 0.012 & 0.017 
  \\
 $\sigma_{\eta}$ 
 & 0.466 & 0.259 & 0.493 & 0.068
 & 3.014 & 1.190 & 0.456 & 0.545 
 & 0.035 & 0.035 & 0.013 & 0.019 
 \\
 $\nu$ 
 & 0.008 & 1.023 & 0.388 & 0.080
 & 2.346 & 1.867 & 0.287 & 0.606 
 & 0.041 & 0.033 & 0.013 & 0.019 
 \\
 $\Ssmhid_{1}$ 
 & 0.498 & 0.184 & 0.389 & 0.096
 & 1.834 & 1.081 & 1.448 & 0.201 
 & 0.067 & 0.033 & 0.013 & 0.019 
 \\
 $\Ssmhid_{5473}$ 
 & 0.571 & 0.152 & 0.385 & 0.053
 & 2.611 & 1.224 & 1.021 & 0.396 
 & 0.035 & 0.031 & 0.009 & 0.014 \\
%		\bottomrule
 \bottomrule
	\end{tabular}
  \end{threeparttable}
\end{table}
%}}}

Like in the Poisson experiment, the overall means of the exact methods
were close to each other, but AI had some bias, this time 
also with some of the hyperparameters ($\sigma_\eta$ and $\nu$). 
The IS1 and IS2 methods outperformed the PM and DA methods similarly
as in the Poisson experiment.
Due to a much smaller $m$, the DA-SPDK and
DA-$\psi$-APF were an order of magnitude faster than DA-BSF.
Diagnostics from the individual runs of PM-SPDK and DA-SPDK sometimes
showed poor mixing, and despite the large reductions in execution
time, the IREs were worse than PM-BSF. We observed also cases with a few
very large correction weights in IS1-SPDK and IS2-SPDK, which had some
impact also on their efficiencies. The SEs of DA-$\psi$-APF were
comparable with the DA-BSF.  We did not experience problems with
mixing or overly large weights with $\psi$-APF, which suggests
$\psi$-APF being more robust than SPDK. There were no significant
differences in the SEs between the exact methods when using the local
and global approximation schemes. 

We calculated the asymptotic variances, the 95\% confidence intervals and 
the variance decomposition as in Section \ref{poisson}.
The coverages of 95\% CIs for IS-MCMC estimators were between 0.91--1
for all variants and variables, except for $\sigma_{\eta}$ for which
the average coverages were 0.77, 0.85, and 0.88 for IS1-SPDK, IS2-BSF,
and IS2-SPDK, but this is likely due to random fluctuation, because
there were only 100 replications. 
For the latent states, the proportions $n\hat{v}_n/\hat{s}_n$ 
were close to 1 in all cases, as in Section \ref{poisson}. For the
hyperparameters, the proportion was around 0.9 for IS1-SPDK, 0.6 for
IS1-$\psi$-APF, 1 for IS2-SPDK, and 0.8 for both IS2-$\psi$-APF and
IS2-BSF. These proportions suggest that SPDK is less efficient than
$\psi$-APF, which is in line with the observed IREs in Table
\ref{table:sv}.

%}}} % SV

%}}}

%%%%%%%%%%%%%%%%%%%%%%%%%%%%%%%%%%%%%%%%
\subsection{Discretely observed Geometric Brownian motion} %{{{

Our last experiment was about a discretely observed diffusion as
discussed in Section \ref{sec:diffusion-ssms}.
The model was a geometric Brownian motion, with
noisy log-observations:
\begin{equation*}
\ud \tilde\Ssmhid_t = \nu \tilde\Ssmhid_t \ud t + \sigma_\ssmhid \tilde\Ssmhid_t
\ud B_t, \qquad \Ssmobs_k \mid (\Ssmhid_k=\ssmhid) \sim N(\log(\ssmhid),
\sigma^2_\ssmobs),
\end{equation*}
with $\tilde\Ssmhid_0\equiv 1$, 
where $(B_t)_{t\ge 1}$ stands for the standard Brownian motion, 
and where 
$\Ssmhid_k = \tilde\Ssmhid_k$.
The discretisations $\mu_t^{(\param)}$ and
$\hat{\mu}_t^{(\param)}$
were based on a Milstein discretisation with uniform meshes
of sizes $2^{L_F}$ and $2^{L_C}$, respectively, with $L_C=4$ and
$L_F=16$, reflected to positive values.
We did not consider optimising $L_C$ and $L_F$, but rather aimed for
illustrating the potential gains for the IS2 algorithm from
parallelisation.
The data was a single simulated realisation of length $50$
from the exact model, with $\nu=0.05$, $\sigma_x = 0.3$, and
$\sigma_y=1$. 
We used a
half-Gaussian prior with s.d.~$0.1$ for $\nu$, a half-Gaussian prior
with s.d.~$0.5$ for $\sigma_x$, and $N(1.5,0.5^2)$ prior truncated to
$>0.5$ for $\sigma_y$. For both IS2 and DA, and both levels, we used $m=50$
which led to $\delta \approx 0.6$.

Assuming a unit cost for each step in the BSF,
the total average cost of a parallel IS2 run is 
$n 2^{L_C} + \alpha (n - n_b) 2^{L_F}/M$,
where $\alpha$ is the mean acceptance rate of the approximate MCMC,
$n_b$ is the length of burn-in and $M$ is the number of
parallel cores used for the weighting. We chose $n =
5000$, $n_b = 2500$, $M=48$, and the target
acceptance rate $\alpha=0.234$, leading to an expected 43-fold
speed-up due to the parallelisation of IS2.

Single run of DA cannot be easily parallelised, but we ran instead
multiple independent DA chains in parallel, and averaged their outputs
for inference. While such parallelisation may not be optimal, it
allows for utilisation of all of the available computational resources.
The running time of each DA chain was constrained to
be similar to the time required by IS2, leading to $n = 100$
with $n_b = 50$. Because of the short runs, we suspected that initial
bias could play a significant role, which was explored 
by running two experiments, with MCMC initialised either 
to the prior mean $\theta_0 = (0.08, 0.4, 1.5)$,
or to an independent sample from the prior. We experimented also
with further thinning,
by forming the IS2 correction based on every other accepted state.

Table \ref{table:gbm-results} summarises the results from 100
replications. The run time of the parallel DA algorithms was defined
as the maximum run time of all parallel chains. The parallelisation
speedup of IS2 was nearly optimal, as well as the further speedup from
thinning.  The SEs with prior mean initialisation were similar between
DA and IS2, but DA produced slightly biased results, leading to 9.5 to
13.0 times higher IREs. The efficiency gains of thinning were
inconclusive, indicating some gains for the hyperparameters $\theta$,
but not for the state variables. The smaller memory requirements and
smaller absolute time requirements for the thinning make it still
appealing. With prior sample initialisation, DA behaved sometimes
poorly, in contrast with IS2 which behaved similarly with both
initialisation strategies.

%%%%%%%%%%%%%%%%%%%%
\begin{table}[t]
	\centering \caption{Results for the geometric Brownian motion
    experiment using 48 cores. IS2\textsuperscript{t} is with
    thinning, and time is in minutes. Ground truth (GT)
    was calculated with MCMC using exact latent inference.}
	\label{table:gbm-results} %{{{
              \small
\begin{threeparttable}
	\begin{tabular}{lccccccccccc}
		\toprule
%      \headrow
    &\multicolumn{6}{c}{Mean}
    & \multicolumn{5}{c}{IRE} \\
    \cmidrule(lr){2-7}
    \cmidrule(lr){8-12}
%    \headrow
    {Init.}
    && \multicolumn{3}{c}{Prior mean}
    & \multicolumn{2}{c}{Prior sample} 
    & \multicolumn{3}{c}{Prior mean}
    & \multicolumn{2}{c}{Prior sample} \\
    \cmidrule(lr){3-5}
    \cmidrule(lr){6-7}
    \cmidrule(lr){8-10}
    \cmidrule(lr){11-12}
%    \headrow
    & GT
		& DA &
    IS2 &  IS2\textsuperscript{t}
    & DA & IS2 
    & DA &
    IS2 &  IS2\textsuperscript{t}
    & DA & IS2 \\
		\midrule
		{Time} & --- &  12.3 & 3.4 & 1.9 & 14.0 & 3.3
    & 12.3 & 3.4 & 1.9 & 14.0 & 3.3\\
%		\midrule
		$\nu$ 		   & 0.053 & 0.061 & 0.053 & 0.053 & 0.064 & 0.053
                         & 0.069 & 0.004 & 0.002 & 0.135 & 0.004 \\ 
		$\sigma_x$   & 0.253 & 0.278 & 0.253 & 0.253 & 0.251 & 0.252
                         & 0.576 & 0.029 & 0.019 & 0.336 & 0.022 \\
		$\sigma_y$	 & 1.058 & 1.054 & 1.058 & 1.058 & 1.083 & 1.058
                         & 0.088 & 0.020 & 0.014 & 1.010 & 0.022 \\
		$\Ssmhid_{1}$& 1.254 & 1.273 & 1.254 & 1.246 & 1.243 & 1.252 
                         & 0.670 & 0.109 & 0.119 & 0.805 & 0.103 \\
		$\Ssmhid_{50}$&2.960 & 2.953 & 2.966 & 2.935 & 20.773 & 2.971
                        & 12.605 & 1.880 & 2.074 & 4$\times 10^6$ & 2.308 \\
		\bottomrule
%  \hline
	\end{tabular}
  \end{threeparttable}
\end{table}
%}}}

%}}}

%%%%%%%%%%%%%%%%%%%%%%%%%%%%%%%%%%%%%%%%
\subsection{Summary of results} %{{{

In our experiments with Laplace approximations, IS1 and IS2 were
competitive alternatives to PM and DA, respectively, even without
parallelisation.  The differences were more
emphasised when the cost of correction (number of samples $m$) 
was higher. The $\psi$-APF was generally preferable
over SPDK, and BSF was the least efficient. 
The global approximation gave additional performance boost in our
experiments, without compromising the accuracy of the estimates, but we
stress that it may not be stable in all scenarios.

As noted earlier, the use of the guidelines by
\cite{doucet-pitt-deligiannidis-kohn} were not necessarily optimal in
our setting. We did an additional experiment to inspect how the choice
of $m$ affects the IRE with BSF in the Poisson model, and
with $\psi$-APF in the SV model. Figure \ref{fig:poisson-ire} shows the
average IREs as a function of $m$. Both IS2 and DA behaved similarly,
and IS2 was less than DA uniformly in terms of IRE. In
the Poisson-BSF case, the choice $m=200$ based on
\cite{doucet-pitt-deligiannidis-kohn} appears nearly optimal.
In case of the SV-$\psi$-APF, the optimal $m$ for DA and IS2 was around 
50, which was higher than $m=10$ based on
\cite{doucet-pitt-deligiannidis-kohn}. This is likely because of the
initial overhead cost of the approximation.

%%%%%%%%%%%%%%%%%%%%
\begin{figure} %{{{
\center
	\includegraphics[width=0.45\textwidth]{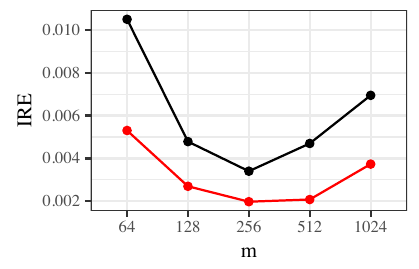}
	\includegraphics[width=0.45\textwidth]{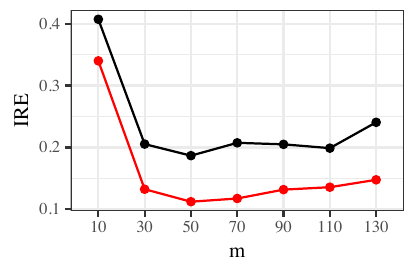}
	\caption{Average IRE of
    $(\sigma_\eta,\sigma_\xi,\Ssmhid_1,\Ssmhid_{100})$ 
    in the Poisson model with BSF (left) and
    of $(\phi,\sigma_\eta, \nu, \Ssmhid_1,\Ssmhid_{5473})$ 
    in the SV model with $\psi$-APF (right).
    DA is shown in black and IS2 in red.} %}}}
	\label{fig:poisson-ire}
\end{figure}

The discretely observed geometric Brownian motion example 
illustrated the potential gains which may be achieved by using the
IS2 method in a parallel environment. While we admit that our
experiment is academic, we believe that it is indicative, and shows
that IS2 can provide substantial gains, and makes reliable
inference possible in a much shorter time than DA. 
The IS framework is less prone to issues 
with burn-in bias, which can be problematic with naive MCMC 
parallelisation based on independent chains.

%}}} % Discussion of results

%}}} Experiments

%%%%%%%%%%%%%%%%%%%%%%%%%%%%%%%%%%%%%%%%%%%%%%%%%%%%%%%%%%%%%%%%%%%%%%%%
\section{Discussion}
\label{sec:discussion} %{{{

Our framework of IS type estimators based on approximate marginal MCMC
provides a general way to construct consistent estimators. Our
experiments demonstrate that the IS estimator can provide substantial
speedup relative to a delayed acceptance (DA) analogue with parallel
computing, and appears to be competitive to DA even without
parallelisation. We believe that IS is often better
than DA in practice, but it is not hard to find simple examples where DA can
be arbitrarily better than IS (and vice versa) \cite{franks-vihola}.
Our followup work \cite{franks-vihola} complements our findings by
theoretical considerations, with guaranteed asymptotic variance bounds
between IS and DA.

IS is known to be difficult to implement efficiently in high
dimensions, but this is not a major concern in most latent variable
models, where the hyperparameters are low-dimensional. It is also generally 
desirable to design the approximate marginal $\pi_a$ to have heavier tails
than the desired marginal $\pi_m$, in order to guarantee bounded
(expected) weights.
When the method of Section \ref{sec:latent} is used,
the IS weight may be directly regularised by
inflating the (estimated) approximate likelihood, for instance with
$L_a(\param)+\epsilon$, with some $\epsilon>0$. If the likelihood $L$
is bounded, then $w_u(\param) \propto
L(\param)/(L_a(\param)+\epsilon)$ is bounded as well. The latter
approach can be seen as an instance of defensive importance sampling
\cite{hesterberg}, using the prior as a proposal component. 
Other generic safe IS schemes may also be useful
\cite[cf.][]{owen-zhou}, and tempering may be applied for the
likelihood as well.

We used adaptive MCMC in order to construct the
marginal chain $(\Param_k)_{k\ge 1}$ in our experiments, and believe
that it is often useful
\cite[cf.][]{andrieu-thoms}. Note, however, that our theoretical results 
do not apply directly with adaptive MCMC, unless the adaptation is stopped
after suitable burn-in. Our results
could be extended to hold with continuous adaptation, under certain
technical conditions.
We detailed proper weighting schemes based on standard IS and particle
filters. We note that various PF variations, such as Rao-Blackwellisation,
alternative resampling strategies \cite{cappe-moulines-ryden},
or quasi-Monte Carlo updates \cite{gerber-chopin}, apply directly.
PFs can also be useful beyond the state space models context
\cite{delmoral-doucet-jasra}. Twisted particle filters
\cite{whiteley-lee,alaluhtala-whiteley-heine-piche} could also be
applied, instead of the $\psi$-APF. 

In a diffusion context, a proper weighting can be constructed based on
randomised multilevel Monte Carlo, as recently described in
 \cite{franks-jasra-law-vihola}, and ABC post-correction may be
seen as IS-type correction \citep{vihola-franks-abc}. 
Laplace approximations are available for a wider class
of Gaussian latent variable models beyond SSMs
\cite[cf.][]{rue-martino-chopin}.  Variational approximations
\cite{beal,jordan-gm} and expectation propagation \cite{minka} have
been found useful in a wide variety of models. In the SSM context,
various non-linear filters could also be applied
\cite[cf.][]{sarkka-filtering-smoothing}. Our framework provides a
generic validation mechanism for approximate inference, where assessment of bias 
is difficult in general \cite[cf.][]{ogden}.
Contrary to purely approximate inference, our approach only requires
moderately accurate approximations, as demonstrated by 
our experiment with global
Laplace approximations.
Debiased MCMC, as suggested in \cite{glynn-rhee} and further explored in
\cite{jacob-lindsten-schon,jacob-oleary-atchade}, may also lead to
useful proper weighting schemes.

%}}}

%%%%%%%%%%%%%%%%%%%%%%%%%%%%%%%%%%%%%%%%%%%%%%%%%%%%%%%%%%%%%%%%%%%%%%%%
\section*{Acknowledgements} %{{{

The authors have been supported by an Academy of Finland research
fellowship (grants 274740, 284513, 312605 and 315619). We thank
Christophe Andrieu, Arnaud Doucet, Anthony Lee and Chris Sherlock
for many insightful remarks.

%}}}

%%%%%%%%%%%%%%%%%%%%%%%%%%%%%%%%%%%%%%%%%%%%%%%%%%%%%%%%%%%%%%%%%%%%%%%%
%%%%%%%%%%%%%%%%%%%%%%%%%%%%%%%%%%%%%%%%%%%%%%%%%%%%%%%%%%%%%%%%%%%%%%%%
\appendix

%%%%%%%%%%%%%%%%%%%%%%%%%%%%%%%%%%%%%%%%%%%%%%%%%%%%%%%%%%%%%%%%%%%%%%%%
\section{Properties of augmented Markov chains}
\label{app:aug} %{{{

Throughout this section, suppose that $K$ is a Markov kernel on $\X$ and 
$Q$ is a kernel from $\X$ to a space $\Sp$. We consider here properties of an augmented
Markov kernel $\check{K}$ defined on $\X\times\Sp$ as
follows:
\begin{align*}
    \check{K}\big((x,s),\ud x'\times \ud s'\big)
    \defeq K(x,\ud x') Q(x',\ud s').
\end{align*}
We first state the following basic result.
%%%%%%%%%%%%%%%%%%%%
\begin{lemma}
    \label{lem:aug-prop} %{{{
          The properties of $K$ and the augmented chain $\check{K}$
          are related as follows:
          \begin{enumerate}[label=(\roman*)]
  \item \label{lem:aug-mim-irr}
    Let $\mathrm{irr}(K)$ denote the set of
    $\phi$-irreducibility measures of a Markov kernel 
    $K$, then
    \begin{itemize}
        \item $\varphi_{K}\in \mathrm{irr}(K) \implies
        \varphi_{\check{K}}(\ud x\times \ud s)
        \defeq \varphi_{K}(\ud x)Q(x,\ud s)
        \in \mathrm{irr}(\check{K})$, 
        \item $\varphi_{\check{K}} \in \mathrm{irr}(\check{K}) \implies
        \varphi_{K}(\ud x)\defeq 
        \varphi_{\check{K}}(\ud x \times \Sp) \in \mathrm{irr}(K)$.
    \end{itemize}
  \item \label{lem:aug-mim-mim}
    The implications in \eqref{lem:aug-mim-irr} hold
    when $\mathrm{irr}(K)$ and $\mathrm{irr}(\check{K})$ are replaced with
    sets of maximal irreducibility measures of $K$ and $\check{K}$,
    respectively.
            \item \label{item:aug-invar}
              The invariant probabilities of $K$ and $\check{K}$ satisfy:
              \begin{itemize}
                  \item $\nu K = \nu
                  \implies \check{\nu} \check{K} =
                  \check{\nu}\quad\text{where}\quad
                  \check{\nu}(\ud x\times\ud y) \defeq \nu(\ud x)
                  Q(x,\ud y)$, 
                  \item $\check{\nu} \check{K} = \check{\nu}
                  \implies \nu K =\nu\quad\text{where}\quad
                  \nu(\ud x)\defeq \check{\nu}(\ud x \times \Sp)$.
              \end{itemize}
              These implications hold also with invariance
              replaced by reversibility.
            \item \label{item:aug-harris} $K$ is Harris recurrent
              if and only if $\check{K}$ is Harris recurrent.
            \item \label{item:aug-iterates}
              Suppose $h:\X\times\Sp\to\R$ is measurable
              and such that $m_h(x) \defeq \int Q(x,\ud s) h(x,s)$
              and $(K^n m_h)(x)$ are well-defined. Then, for any $n\ge 1$
              and $s\in\Sp$, 
              $(\check{K}^n h)(x,s) = (K^n m_h)(x)$.
          \end{enumerate}
      \end{lemma} %}}}
\begin{proof} %{{{
The inheritance of irreducibility measures \ref{lem:aug-mim-irr},
maximal irreducibility measures \ref{lem:aug-mim-mim}, invariant
measures \ref{item:aug-invar}, and reversibility is straightforward.

          For Harris recurrence \ref{item:aug-harris},
          let the probability $\phi_K$ be a maximal
          irreducibility measure for $K$, then
          $\phi_{\check{K}}(\ud x\times
          \ud s) \defeq \phi_K(\ud x) Q(x,\ud s)$
          is the maximal irreducibility measure for $\check{K}$.
          Let $C\in \cX\otimes\cSp$ with
          $\phi_{\check{K}}(C)>0$, and choose $\epsilon>0$ such that
          $\phi_K(C(\epsilon))>0$, where $C(\epsilon) \defeq \{x\in\X\given
            Q(x,C_x)>\epsilon\}$ with $C_x \defeq \{s\in \Sp\given (x,s)\in C\}$.
          Notice that
          \begin{align*}
              \P\bigg(\sum_{k=1}^\infty \charfun{(X_k,S_k)\in C} =\infty \bigg)
              &\ge \P\bigg(\sum_{k=1}^\infty \charfun{S_{\tau_k}\in
                C_{X_{\tau_k}}}=\infty \bigg),
          \end{align*}
          where $\tau_k$ are the hitting times of $(X_k)$ to $C(\epsilon)$.
          This concludes the proof because
          $\charfun{S_{\tau_k} \in C_{X_{\tau_k}}}$ are independent
          Bernoulli random variables with success probability at least
          $\epsilon$.  The converse statement is similar.

          For \ref{item:aug-iterates}, it is enough to notice that
          for any $(x,s)\in\X\times\Sp$ and $n\ge 1$, it holds that
          $\check{K}^n\big((x,s),\ud x'\times\ud s'\big) = K^n(x,\ud
          x')Q(x',\ud s)$.
      \end{proof} %}}}

We next state the following generic results about
the asymptotic variance and the central limit theorem
of an augmented Markov chain.
For $h\in L_0^2(\check{\nu})$, we denote as above
the conditional mean $m_h(x) \defeq \int Q(x,\ud s)h(x,s)$
and the conditional variance
$v_h(x) \defeq \int Q(x,\ud s)h^2(x,s) - m_h^2(x)$.
%%%%%%%%%%%%%%%%%%%%
\begin{lemma}
    \label{lem:aug-asvar} %{{{
          Let $h\in L_0^2(\check{\nu})$.
          The asymptotic variance of an augmented Markov chain satisfies
          \[
              \Var(h,\check{K}) = \Var(m_h, K) + \nu(v_h),
          \]
          whenever $\Var(m_h, K)$ is well-defined.
      \end{lemma} %}}}
\begin{proof} %{{{
Let $(X_k,S_k)$ be a stationary Markov chain with transition
probability $\check{K}$.
\begin{align*}
    \Var\bigg(\frac{1}{\sqrt{n}} \sum_{k=1}^n h(X_k,S_k)\bigg)
    &= \check{\nu}(h^2) + \frac{2}{n} \sum_{i=1}^{n-1}
    \sum_{\ell=1}^{n-i} \E[h(X_0,S_0)h(X_\ell,S_\ell)],
\end{align*}
by stationarity. For $\ell\ge 1$, Lemma \ref{lem:aug-prop}
\ref{item:aug-iterates} implies
\begin{align*}
    \E[h(X_0,S_0)h(X_\ell,S_\ell)]
    &= \E[m_h(X_0) m_h(X_\ell)].
\end{align*}
We deduce for any $n\ge 1$
\[
\Var\bigg(\frac{1}{\sqrt{n}} \sum_{k=1}^n h(X_k,S_k)\bigg)
    = \Var\bigg(\frac{1}{\sqrt{n}} \sum_{k=1}^n m_h(X_k) \bigg)
    + \nu(v_h),
\]
because
$\check{\nu}(h^2) - \nu(m_h^2) = \nu(v_h)$. The claim
follows by taking limit $n\to\infty$.
\end{proof} %}}}

%%%%%%%%%%%%%%%%%%%%
\begin{lemma}
    \label{lem:aug-clt} %{{{
Suppose $K$ is Harris ergodic and aperiodic, and $h\in L_0^2(\check{\nu})$.
The CLT
\begin{equation}
    Z_n \defeq 
          \frac{1}{\sqrt{n}} \sum_{k=1}^n h(X_k,S_k)
          \xrightarrow{n\to\infty} N\big(0,\sigma_h^2\big),\qquad 
          \text{where}\qquad
          \sigma_h^2 \defeq \Var(m_h, K) + \nu(v_h)
          \label{eq:aug-clt}
\end{equation}
holds for every initial distribution, if one of the following holds:
\begin{enumerate}[label=(\roman*)]
    \item
      \label{item:kipnis-varadhan}
      $K$ is reversible and $\Var(m_h, K)<\infty$.
    \item
      \label{item:maxwell-woodroofe}
      $\sum_{n=1}^\infty n^{-3/2} \big\{\nu\big( \big[
        \sum_{i=0}^{n-1}
      K^i m_h\big]^2\big)\big\}^{1/2}<\infty$.
    \item
      \label{item:poisson-clt}
      There exists $g\in L^2(\nu)$ which solves the Poisson equation
      $g - Kg = m_h$.
      In this case, $\Var(m_h,K) = \nu(g^2 - (Kg)^2)$.
\end{enumerate}
\end{lemma} %}}}
\begin{proof} %{{{
The case \ref{item:kipnis-varadhan}
follows from Lemma \ref{lem:aug-asvar} and the Kipnis-Varadhan CLT
\citep{kipnis-varadhan}, which implies \eqref{eq:aug-clt} for the
initial distribution $\check{\nu}$.
Because the jump chain is Harris by Lemma \ref{lem:aug-prop}
\ref{item:aug-harris}, \eqref{eq:aug-clt} holds for every initial distribution
\cite[Corollary 21.1.6]{douc-moulines-priouret-soulier}.

The case \ref{item:maxwell-woodroofe} follows similarly, but
relies on a result due to Maxwell and Woodroofe
\cite{maxwell-woodroofe}, which implies \eqref{eq:aug-clt}
for the initial distribution $\check{\nu}$, if
$\sum_{n=1}^\infty 
n^{-3/2} \big\{ \check{\nu}\big(\big[\sum_{i=0}^{n-1}
    \check{K}^i h \big]^2\big)\big\}^{1/2}<\infty$.
Notice that for $n\ge 2$ by Lemma \ref{lem:aug-prop}
\ref{item:aug-iterates},
\begin{align*}
\check{\nu}\bigg( \bigg[
        \sum_{i=0}^{n-1}
      \check{K}^i h\bigg]^2\bigg)
    &=
    \check{\nu}\bigg( \bigg[ (h - m_h) +
        \sum_{i=0}^{n-1}
      K^i m_h\bigg]^2\bigg)
    = \nu(v_h) +
     \nu\bigg( \bigg[
        \sum_{i=0}^{n-1}
      K^i m_h\bigg]^2\bigg).
\end{align*}
Because $(a+b)^{1/2} \le a^{1/2}+b^{1/2}$ for $a,b\ge 0$
and $\nu(v_h)<\infty$, the claim follows.

For \ref{item:poisson-clt}, we first observe that
\[
    \check{g} - \check{K}\check{g} = h
    \qquad\text{where}\qquad
    \check{g}(x,s) \defeq g(x) + h(x,s)-m_h(x)\in L^2(\check{\nu}).
\]
Indeed, it is clear that $\check{g}\in L^2(\check{\nu})$ and because
$(\check{K} \check{g})(x,s) = (Kg)(x)$,
\[
    \check{g}(x,s) - (\check{K} \check{g})(x,s)
    = g(x) - (Kg)(x) + h(x,s) - m_h(x) = h(x,s).
\]
The CLT and asymptotic variance follow from
\cite[Theorem 17.4.4]{meyn-tweedie}.
\end{proof} %}}}

%}}}

%%%%%%%%%%%%%%%%%%%%%%%%%%%%%%%%%%%%%%%%%%%%%%%%%%%%%%%%%%%%%%%%%%%%%%%%
\section{Proofs about CLT and asymptotic variance}
\label{app:clt} %{{{

%%%%%%%%%%%%%%%%%%%%
\begin{proof}[Proof of Theorem \ref{thm:proper-clt-rev}] %{{{
Whenever $\sum_{i=1}^n \proper_i(\unitfun) >0$, we may write
\begin{equation}
    \sqrt{n}\big[E_n(f)-\jointdens(f)\big]
    = \frac{n^{-1/2} \sum_{k=1}^n \proper_k(\bar{f}) }{
      n^{-1} \sum_{j=1}^n \proper_j(\unitfun)}.
    \label{eq:unbiased-is-ratio}
\end{equation}
The denominator converges to $c_w>0$ almost surely, so by Slutsky's
lemma, it is enough to show that the numerator converges in distribution to
$N\big(0,\Var(\nu_{\bar{f}}, P) + \approxdens(v)\big)$.
This follows from Lemma \ref{lem:aug-clt} \ref{item:kipnis-varadhan}
and \ref{item:maxwell-woodroofe}, under conditions \ref{item:proper-kv}
and \ref{item:proper-mw}, respectively.
\end{proof}
%}}}

%%%%%%%%%%%%%%%%%%%%
\begin{proof}[Proof of Corollary \ref{cor:is-ci}] %{{{
If $\sqrt{n}\big[E_n(f)-\jointdens(f)\big]\to
N(0,\sigma_f^2)$, Slutsky's lemma applied to 
\eqref{eq:unbiased-is-ratio} implies that $n^{-1/2} \sum_{k=1}^n
\proper_k(\bar{f}) \to N(0, c_w^2 \sigma_f^2)$. Consistency of the
integrated autocovariance estimator implies 
$\hat{a}_n \to \sum_{k=-\infty}^\infty \gamma_k 
= c_w^2 \sigma_f^2$, and therefore $\hat{s}_n \to \sigma_f^2$.
The conclusion follows by a standard continuity argument.
\end{proof}
%}}}

%%%%%%%%%%%%%%%%%%%%
\begin{proof}[Proof of Theorem \ref{thm:importance-var}] %{{{
For $n$ large enough such that $\sum_{j=1}^n \proper_j(\unitfun)>0$,
we may write
\[
        n \hat{v}_n
   = \frac{\frac{1}{n} \sum_{k=1}^n \big(\proper_k(f)
     - \proper_k(\unitfun) E_n(f) \big)^2}{\big(\frac{1}{n}\sum_{j=1}^n
     \proper_j(\unitfun)\big)^2}.
\]
The denominator converges to $c_w^2$, and the
numerator can be written as
\[
    \frac{1}{n} \sum_{k=1}^n \big[\proper_k^2(\bar{f})
    + \proper_k^2(\unitfun)D_n^2 + 2
    \proper_k(\unitfun)\proper_k(\bar{f})D_n \big]
    \qquad\text{with}\qquad
    D_n \defeq \jointdens(f)-E_n(f).
\]
  The term $n^{-1} \sum_{k=1}^n \proper_k^2(\bar{f})\to
\approxdens(v+\meanfunc_{\bar{f}}^2)$, and because $D_n\to 0$,
the remainder terms tend to zero.
\end{proof} %}}}

%}}}

%%%%%%%%%%%%%%%%%%%%%%%%%%%%%%%%%%%%%%%%%%%%%%%%%%%%%%%%%%%%%%%%%%%%%%%%
\section{Proofs about jump chain estimators}
\label{app:jump} %{{{

In this section, $K$ is assumed to be a Markov kernel
on $\X$ which is non-degenerate, that is, 
$a(x) \defeq K(x,\X\setminus\{x\}) >0$ for all $x\in\X$.
The following proposition complements
\cite[Lemma 1]{douc-robert} and \cite{deligiannidis-lee}, which are
stated for more specific cases.
%%%%%%%%%%%%%%%%%%%%
\begin{proposition}
    \label{prop:jump-properties} %{{{
          Suppose $(X_k)$ is a Markov chain with
          kernel $K$ and
          $(\tilde{X}_k)$ the corresponding jump chain 
          with holding
          times $(N_k)$ (Definition \ref{def:jump}).
          Then, the following hold:
          \begin{enumerate}[label=(\roman*)]
            \item
              \label{item:jump-trans}
              $(\tilde{X}_k)$ is Markov with transition
              kernel
              $\tilde{K}(x,A) = K(x,A\setminus\{x\})/a(x)$.
            \item
              \label{item:jump-joint-trans}
              The holding times $(N_k)$ are conditionally independent
              given $(\tilde{X}_k)$, and each $N_k$ has geometric
              distribution with parameter $a(\tilde{X}_k)$.
            \item
              \label{item:jump-invar}
              If $K$ admits invariant probability $\nu(\ud x)$, then
              $\tilde{K}$ admits invariant probability
              $
                  \tilde{\nu}(\ud x) \defeq \nu(\ud x) a(x)/\nu(a).
              $
              In addition, if $K$ is reversible with respect to $\nu$,
              then $\tilde{K}$ is reversible with respect to $\tilde{\nu}$.
            \item
              \label{item:psi-irreducibility} $(X_k)$ is $\psi$-irreducible
              if and only if $(\tilde{X}_k)$ is
              $\psi$-irreducible, with the same maximal irreducibility measure.
            \item \label{item:harris-recurrence}
              $(X_k)$ is Harris recurrent if and only if
              $(\tilde{X}_k)$ is Harris recurrent.
          \end{enumerate}
      \end{proposition} %}}}
\begin{proof} %{{{
          The expression of the kernel \ref{item:jump-trans}
          is due to straightforward conditioning, and \ref{item:jump-joint-trans}
          was observed in \cite{douc-robert}. The invariance \ref{item:jump-invar}
          follows from
          \begin{align*}
              \int \tilde{\nu}(\ud x) \tilde{K}(x,A)
              &= \frac{1}{\nu(a)} \int \nu(\ud x) \big[K(x,A) -
              \charfun{x\in A} K(x,\{x\})\big] \\
              &
              = \frac{1}{\nu(a)}\bigg[ \nu(A) - \int_A \nu(\ud
              x)\big(1-a(x)\big)\bigg]
              = \tilde{\nu}(A),
          \end{align*}
          and the reversibility is shown in
          \cite{douc-robert}.
          For \ref{item:psi-irreducibility} it is sufficient to observe that
          \[
              \forall x\in \X: \sum_{n\ge 1} \P_x(X_n\in A)>0
              \iff
              \forall x\in \X: \sum_{n\ge 1} \P_x(\tilde{X}_n\in A)>0,
          \]
          where $\P_x(\uarg) = \P(\uarg\mid X_0=\tilde{X}_0=x)$,
          which holds because the sets $\{X_k\}_{k\ge 0}$ and
          $\{\tilde{X}_k\}_{k\ge 0}$ coincide. Similarly,
          \ref{item:harris-recurrence} holds because
          \[
              \forall x\in \X:
              \P_x(\eta_A =\infty) = 1 \iff
              \forall x\in\X:
              \P_x(\tilde{\eta}_A =\infty) = 1,
          \]
          where $\eta_A \defeq \sum_{k=1}^\infty \charfun{X_k\in A}$ and
          $\tilde{\eta}_A \defeq \sum_{k=1}^\infty
          \charfun{\smash{\tilde{X}_k\in A}}$.
      \end{proof} %}}}

We now state results about the asymptotic variance of the jump chain,
complementing the reversible case characterisation of 
\cite{deligiannidis-lee,doucet-pitt-deligiannidis-kohn}.
%%%%%%%%%%%%%%%%%%%%
\begin{proposition}
    \label{prop:jump-variance} %{{{
          Let $f\in L^2_0(\tilde{\nu})$.
          With the notation of Proposition \ref{prop:jump-properties},
          \begin{enumerate}[label=(\roman*)]
              \item
                \label{item:deligiannidis-lee}
                If $K$ is
                reversible, then $\Var(f,\tilde{K})<\infty$ iff
                $af\in L^2(\nu)$ and
                $\Var(af,K)<\infty$, and
          \begin{equation}
              \Var(f,\tilde{K}) = \nu(a)^{-1}
              \big[ \Var(a f, K) - \nu\big(a(1-a) f^2\big)\big].
              \label{eq:jump-variance}
          \end{equation}
              \item
                \label{item:jump-poisson}
                If there exists
          a function $g\in L^2(\nu)$ which satisfies $g - Kg = af$,
          then $\Var(f,\tilde{K})<\infty$, $\Var(af,K)<\infty$,
          \eqref{eq:jump-variance} holds,
          $g - \tilde{K} g = f$ and $g\in L^2(\tilde{\nu})$.
          \end{enumerate}
      \end{proposition} %}}}
\begin{proof} %{{{
          The reversible case \ref{item:deligiannidis-lee} is a
          restatement of \cite[Theorem 1]{deligiannidis-lee}.

          Consider then \ref{item:jump-poisson}. By
          Proposition \ref{prop:jump-properties}
          \ref{item:jump-trans}, we obtain
          for any $h:\X\to\R$ with $Kh$ well-defined,
          \[
              (\tilde{K} h)(x) = \frac{(Kh)(x) -
                \big(1-a(x)\big)h(x)}{a(x)}
              = \frac{(Kh)(x) -h(x)}{a(x)} + h(x).
          \]
          Consequently, we observe that
              $g - \tilde{K} g = a^{-1} \big( g - Kg\big) =
              f$
          implying \ref{item:jump-poisson}.
          Because $g\in L^2(\tilde{\nu})$, Lemma \ref{lem:aug-clt}
          \ref{item:poisson-clt}
          and a straightforward calculation yield
          \begin{align*}
              \Var(f, \tilde{K}) &= \tilde{\nu}\big(g^2 - (\tilde{K}g)^2\big) 
              \\
              &
              = 2 \langle g, g - \tilde{K} g\rangle_{\tilde{\nu}}
              - \langle g-\tilde{K}g, g -
              \tilde{K}g\rangle_{\tilde{\nu}} \\
              &
              = \nu(a)^{-1}  \big[
              2 \langle g, g - Kg \rangle_{\nu} - \nu( a f^2)
              \big],
          \end{align*}
          where $\langle f,g\rangle_{\nu} \defeq \int f(x) g(x)
          \nu(\ud x)$.
          Similarly, by Lemma  \ref{lem:aug-clt}
          \ref{item:poisson-clt}
          \[
              \Var(a f, K) = \nu\big(g^2 -
              (Kg)^2\big) = 2 \langle g, g - Kg \rangle_\nu - \nu( a^2 f^2),
          \]
          which allows us to conclude.
      \end{proof} %}}}

%%%%%%%%%%%%%%%%%%%%
\begin{proof}[Proof of Theorem \ref{thm:block-clt}] %{{{
Whenever $\sum_{j=1}^n \proper_j(\unitfun)>0$, we may write
\begin{equation*}
    \sqrt{n}\big[E_n(f)-\jointdens(f)\big] = \frac{n^{-1/2}
      \sum_{k=1}^n N_k \proper_k(\bar{f})  }{n^{-1}
      \sum_{j=1}^n N_j \proper_j(\unitfun)}.
\end{equation*}
We shall show below that
the CLT holds for the numerator, with asymptotic variance
$\sigma^2 \defeq \big[\Var(\meanfunc_{\bar{f}}, P)+\approxdens(\alpha
\tilde{v})\big]/ \approxdens(\alpha)$.
This implies the claim by Slutsky's lemma, as the denominator converges to
$c_w/\approxdens(\alpha)$. For the rest of the proof, let
$\tilde{P}$ and $\check{P}$ be the Markov kernels 
of $(\tilde{\Param}_k)_{k\ge 1}$ and
$(\tilde{\Param}_k$,%
$N_k$,%
$\proper_k(\bar{f}))_{k\ge 1}$,
respectively, and let $\tilde{\pi}$ and
$\check{\pi}$ be the corresponding invariant
probabilities. Note that the function
$h(\param,n,\proper) \defeq n \proper$ is in $L^2(\check{\pi})$
by assumption \eqref{eq:clt-tight}.

In case \ref{item:block-clt-rev} holds, also
$\tilde{P}$ and $\check{P}$ are reversible
by Proposition \ref{prop:jump-properties} \ref{item:jump-invar}
and Lemma \ref{lem:aug-prop} \ref{item:aug-invar}.
Lemma \ref{lem:aug-clt} \ref{item:kipnis-varadhan}
with $K = \tilde{P}$, $\check{K}=\check{P}$, $\nu =
\tilde{\pi}$ and $\check{\nu}=\check{\pi}$ implies
that a CLT holds for $h$ whenever the asymptotic variance is finite:
\[
    \Var(h,\check{P})
    = \Var\big(\meanfunc_{\bar{f}}/\alpha, \tilde{P}\big)
    + \approxdens(\alpha \tilde{v}_{N\proper})/\approxdens(\alpha),
\]
where, by the variance decomposition formula,
\begin{align*}
    \tilde{v}_{N\proper}(\param) &\defeq \Var(N_k \proper_k(\bar{f}) \mid
    \tilde{\Param}_k=\param) 
    \\
    &
    = \tilde{v}(\param) + \Var\big(N_k \E[\proper_k(\bar{f}) \mid
    \tilde{\Param}_k=\param, N_k]\bigmid \tilde{\Param}_k=\param\big) 
    \\
    &
    = \tilde{v}(\param) + \meanfunc_{\bar{f}}^2(\param)
    \big(1-\alpha(\param)\big)/\alpha^2(\param).
\end{align*}
Proposition \ref{prop:jump-variance} \ref{item:deligiannidis-lee}
implies that
\[
    \Var\big(\meanfunc_{\bar{f}}/\alpha, \tilde{P}\big)
    = \approxdens(\alpha)^{-1}
    \big[ \Var(\meanfunc_{\bar{f}}, P) 
    - \approxdens\big((1-\alpha)
      \meanfunc_{\bar{f}}^2/\alpha\big)\big],
\]
which implies $\Var(h,\check{P}) = \sigma^2$.

Consider then \ref{item:block-clt-poisson}.
Proposition \ref{prop:jump-variance}
\ref{item:jump-poisson} implies that
$g - \tilde{P} g = \meanfunc_{\bar{f}}/\alpha$, and $g\in L^2(\tilde{\pi})$.
Lemma \ref{lem:aug-clt} \ref{item:poisson-clt} implies the CLT,
and together with Proposition \ref{prop:jump-variance}
\ref{item:jump-poisson} leads to $\Var(h,\check{P})=\sigma^2$.
\end{proof} %}}}

%}}}

%%%%%%%%%%%%%%%%%%%%%%%%%%%%%%%%%%%%%%%%%%%%%%%%%%%%%%%%%%%%%%%%%%%%%%%%
\section{Proper weightings for general state space models}
\label{sec:ssm} %{{{

We review some techniques to construct proper weightings in case of
general state-space models introduced in Section
\ref{sec:lin-gauss-ssms}. First, note that the simple IS correction 
may be applied directly (see Proposition \ref{prop:augmented-is}).
Note that \eqref{eq:ssm-proper} is satisfied for all integrable $h$,
so $\Ls = L^1(\jointdens)$. It is often useful to combine such schemes
as in Proposition \ref{prop:convex-proper}, allowing for instance
variance reduction by using pairs of antithetic variables
\cite{durbin-koopman2000}.

For the rest of the section, we focus on the particle filter (PF)
algorithm \cite{gordon-salmond-smith}; see also the monographs
\cite{doucet-freitas-gordon,del-moral,cappe-moulines-ryden}. We
consider a generic version of the algorithm, with the following
components \cite[cf.][]{del-moral}:
\begin{enumerate}
    \item Proposal distributions: $M_1$
      is a probability density on $\Sp_\ssmhid$
      and $M_t(\uarg\mid \ssmhid_{1:t-1})$ defines conditional
      densities on $\Sp_\ssmhid$ given
      $\ssmhid_{1:t-1}\in\Sp_\ssmhid^{t-1}$.
    \item Potential functions: $G_t:\Sp_\ssmhid^t\to\R_+$.
    \item
      \label{item:resampling}
      Resampling laws:
    $\mathrm{Res}(\uarg\mid \bar{\omega}^{(1:m)})$
    defines a probability distribution on $\{1{:}m\}^m$
    for every discrete probability mass
    $\bar{\omega}^{(1:m)}$.
\end{enumerate}

The following well-known two conditions are minimal to ensure unbiasedness, which
is required for proper weighting:
%%%%%%%%%%%%%%%%%%%%
\begin{assumption}
   \label{a:smc-consistency} %{{{
Suppose that the following hold:
\begin{enumerate}[label=(\roman*)]
    \item \label{a:fk-proper}
$\prod_{t=1}^T
    M_t(\ssmhid_t\mid \ssmhid_{1:t-1})G_t(\ssmhid_{1:t}) \
    = p^{(\param)}(\ssmhid_{1:T},y_{1:T})$ for all
    $\ssmhid_{1:T}\in\Sp_\ssmhid^T$.
    \item \label{a:resampling}
$\E\big[\sum_{i=1}^m
    \charfun{\vphantom{(}\smash{A^{(i)}=j}}\big]=m\bar{\omega}^{(j)}$,
    where $A^{(1:m)} \sim \mathrm{Res}(\uarg\mid \bar{\omega}^{(1:m)})$,
for any $j\in\{1{:}m\}$ and any probability mass vector $\bar{\omega}^{(1:m)}$.
\end{enumerate}
\end{assumption} %}}}
Assumption \ref{a:smc-consistency}
\ref{a:fk-proper} holds
with traditionally used `filtering' potentials
$G_t(\ssmhid_{1:t})
   \defeq g_t^{(\param)}(\ssmobs_t\mid \ssmhid_t)
            \mu_t^{(\param)}(\ssmhid_t\mid
            \ssmhid_{t-1})/
            M_t(\ssmhid_t\mid
            \ssmhid_{1:t-1})$,
assuming a suitable support condition.
We discuss another choice of $M_t$ and
$G_t$ in Section \ref{sec:lin-gauss-ssms}, 
inspired by the `twisted SSM' approach of \cite{guarniero-johansen-lee}.
It allows a `look-ahead' strategy based on 
approximations of the full smoothing
distributions $q^{(\param)}(\ssmhid_{1:T}\mid \ssmobs_{1:T})$.
Assumption \ref{a:smc-consistency}
\ref{a:resampling} allows for multinomial
resampling, where $A_t^{(i)}$ are independent draws from
$\bar{\omega}_t^{(1:m)}$, but also for 
lower variance schemes, including stratified,
residual and systematic resampling methods
\cite[cf.][]{douc-cappe-moulines}.

Below, whenever
the index `$i$' appears, it takes values
$i=1,\ldots,m$.
%%%%%%%%%%%%%%%%%%%%
\begin{algorithm}[Particle filter]
\label{alg:pf} %{{{
Initial state:
      \begin{enumerate}[label=(\roman*)]
        \item Sample $\Ssmhid_1^{(i)}\sim M_1$ and
          set $\bar{\Ssmhid}_1^{(i)} = \Ssmhid_1^{(i)}$.
        \item Calculate $\omega_1^{(i)} \defeq G_1(\Ssmhid_1^{(i)})$
          and set $\bar{\omega}_1^{(i)} \defeq \omega_1^{(i)}/\omega_1^*$
          where $\omega_1^* = \sum_{j=1}^m
          \omega_{1}^{(j)}$.
      \end{enumerate}
      For $t=2,\ldots,T$, do:
      \begin{enumerate}[label=(\roman*)]
          \stepcounter{enumi}
          \stepcounter{enumi}
        \item Sample $A_{t-1}^{(1:m)} \sim
          \mathrm{Res}(\uarg\mid \bar{\omega}_{t-1}^{(1:m)})$.
        \item Sample $\Ssmhid_t^{(i)}\sim
          M_t(\uarg\mid \bar{\Ssmhid}_{t-1}^{(A_{t-1}^{(i)})})$
          and set $\bar{\Ssmhid}_{t}^{(i)} = (
          \bar{\Ssmhid}_{t-1}^{(A_{t-1}^{(i)})}, \Ssmhid_t^{(i)})$.
        \item
          Calculate $\omega_t^{(i)} \defeq
          G_t(\bar{\Ssmhid}_{t-1}^{(A_{t-1}^{(i)})}, \Ssmhid_t)$
          and set $\bar{\omega}_t^{(i)} \defeq \omega_t^{(i)}/\omega_t^*$
          where $\omega_t^* = \sum_{j=1}^m
          \omega_{t}^{(j)}$.
      \end{enumerate}
\end{algorithm} %}}}

%%%%%%%%%%%%%%%%%%%%
\begin{remark2} %{{{
If all weights are zero, $\omega_{t}^*=0$, then Algorithm \ref{alg:pf}
may be terminated immediately (cf.~Proposition
\ref{prop:particle-proper}).
\end{remark2} %}}}

The following result summarises alternative ways how the random
variables $(\properweightalt^{(1:m)}_\param,\Latent^{(1:m)}_\param)$ may be
constructed from the PF output, in order to satisfy
\eqref{eq:ssm-proper}. The results stated below are gathered
from the literature \cite[e.g.][]{del-moral,pitt-dossantossilva-giordani}, and 
some may be stated under slightly more stringent 
conditions; a self-contained proof of
Proposition \ref{prop:particle-proper} may be found, for instance, in 
\cite{vihola-helske-franks-preprint}.

%%%%%%%%%%%%%%%%%%%%
\begin{proposition}
\label{prop:particle-proper} %{{{
Let $\param\in\paramspace$ be fixed, assume $\mathrm{Res}$,
$M_t$ and $G_t$ satisfy Assumption \ref{a:smc-consistency}, and let
$h:\Sp_\ssmhid^T\to\R$ be such that the integral in
\eqref{eq:ssm-proper} is well-defined and finite.
Consider the random variables generated by Algorithm
\ref{alg:pf}, and let $U\defeq \prod_{t=1}^T
\big(\frac{1}{m}\omega_t^*\big)$. Then,
\begin{enumerate}[label=(\roman*)]
\item \label{item:filter-smoother}
  the random variables
  $(\properweightalt^{(1:m)}_\param,\Latent^{(1:m)}_\param)$ where $
  \properweightalt^{(i)}_\param =
  U \bar{\omega}_T^{(i)}$ and $\Latent^{(i)}_\param = \bar{\Ssmhid}_T^{(i)}$
  satisfy \eqref{eq:ssm-proper}.
\end{enumerate}
Suppose in addition that $M_t(\ssmhid_t\mid \ssmhid_{1:t-1})
G_t(\ssmhid_{1:t}) = C_t(\ssmhid_{t-1:t})$
for all $t\in\{1{:}T\}$ and all $\ssmhid_{1:T}\in\Sp_\ssmhid^T$.
Define for $t\in\{2{:}T\}$, and any $i_t,i_{t-1}\in\{1{:}m\}$,
the backwards sampling probabilities
\[
    b_{t-1}(i_{t-1}\mid i_t)
    \defeq \frac{
      \bar{\omega}_{t-1}^{(i_{t-1})}
      C_t(\Ssmhid_{t-1}^{(i_{t-1})}, \Ssmhid_t^{(i_t)})
      }
{
  \sum_{\ell=1}^m
      \bar{\omega}_{t-1}^{(\ell )}
      C_t(\Ssmhid_{t-1}^{(\ell)}, \Ssmhid_t^{(i_t)})
},
\quad\text{and}\quad
b_T(i_T\mid i_{T+1})
=\bar{\omega}_T^{(i_T)}.
\]
\begin{enumerate}[label=(\roman*)]
    \stepcounter{enumi}
    \item \label{item:backwards-sampling}
      Let $I_{1:T}$ be random indices generated recursively backwards by $I_T\sim
b_T$ and $I_{t}\sim b_t(\uarg\mid I_{t+1})$.
The random variables
$(\properweightalt^{(1)}_\param, \Latent^{(1)}_\param)$
satisfy \eqref{eq:ssm-proper}, where
$\properweightalt^{(1)}_\param=U$ and $\Latent^{(1)}_\param =
\Ssmhid_{1:T}^{(I_{1:T})}$.
    \item \label{item:fwd-bwd-smoothing}
      If $h(\ssmhid_{1:T}) =
      \hat{h}(\ssmhid_{t-1},\ssmhid_{t})$ for some $t\in \{2{:}T\}$,
      that is, $h$
      is constant in all coordinates except $t-1$ and $t$, then,
      the random variables
      $(\properweightalt^{(1:m,1:m)}_\param$, $\Latent^{(1:m,1:m)}_\param)$ satisfy \eqref{eq:ssm-proper}
      (with $\hat{h}$ on the left),
      where
      \begin{enumerate}
      \item $\Latent^{(i,j)}_\param \defeq
        (\Ssmhid_{t-1}^{(i)},\Ssmhid_{t}^{(j)})$,
      \item $\properweightalt^{(i,j)}_\param
        \defeq U b_{t-1}(i\mid j)\hat{\omega}_{t}^{(j)}$,
        and where
      \item \label{item:smoothing-weights}
      $\hat{\omega}_T^{(i)} \defeq
      \bar{\omega}_T^{(i)}$ and $\hat{\omega}_{t}^{(i)} \defeq \sum_{k=1}^m
          \hat{\omega}_{t+1}^{(k)} b_{t}(i\mid k)$ for $t=T-1,\ldots,t$.
   \end{enumerate}
   \item \label{item:fwd-bwd-smoothing2}
     If $h(\ssmhid_{1:T}) = \hat{h}(\ssmhid_t)$ for some $t\in
      \{1{:}T\}$, then the random variables
      $(\properweightalt^{(1:m)}_\param,\Latent^{(1:m)}_\param)$ satisfy
      \eqref{eq:ssm-proper} (with $\hat{h}$ on the left), where
      $\Latent^{(i)}_\param = \Ssmhid_t^{(i)}$ and
      $\properweightalt^{(i)}_\param =
      U \hat{\omega}_t^{(i)}$ are defined in
      \ref{item:smoothing-weights}.
\end{enumerate}
\end{proposition} %}}}

The estimator in Proposition \ref{prop:particle-proper}
\ref{item:filter-smoother} was called the filter-smoother in
\cite{kitagawa}. This property was shown in \cite[Theorem
7.4.2]{del-moral} in case of multinomial resampling, and extended later 
\cite[cf.][]{andrieu-doucet-holenstein}. 
Proposition \ref{prop:particle-proper} \ref{item:backwards-sampling}
corresponds to backwards simulation smoothing
\cite{godsill-doucet-west}.
Drawing a single backward trajectory does not improve on the
filter-smoother 
\cite[cf.][]{doucet-lee-smc-gm}, but
drawing several $I_{1:T}$ 
independently may lead to lower variance estimators.
Proposition \ref{prop:particle-proper}
\ref{item:fwd-bwd-smoothing} and its special case
\ref{item:fwd-bwd-smoothing2} correspond to the
forward-backward smoother \cite{doucet-godsill-andrieu}; see
also \cite{cappe-moulines-ryden}. It is a
Rao-Blackwellised version of \ref{item:backwards-sampling}, but
applicable only when considering estimates of a single marginal (pair).
This scheme can lead
to lower variance, but suffers from $O(m^2)$ complexity.

We next formally state how Proposition
\ref{prop:particle-proper} allows to use Algorithm
\ref{alg:pf} to derive a proper weighting scheme.

%%%%%%%%%%%%%%%%%%%%
\begin{corollary}
    \label{cor:proper-ssm} %{{{
Let $(\Param_k)_{k\ge 1}$ be a Markov chain which is Harris ergodic
with respect to $\approxdens$. Suppose each
$(\properweightalt^{(1:m)}_k,\Latent^{(1:m)}_k)$ corresponds to
an independent run of
Algorithm \ref{alg:pf} with $\param=\Param_k$,
as defined in Proposition \ref{prop:particle-proper}
\ref{item:filter-smoother}, \ref{item:backwards-sampling},
\ref{item:fwd-bwd-smoothing} or
\ref{item:fwd-bwd-smoothing2}.
Then,
$(\properweight_k^{(1:m)}, \Latent^{(1:m)}_k)_{k\ge 1}$ with
$\properweight_k^{(i)} \defeq \pr(\param_k) \properweightalt_k^{(i)}/\approxdens(\param_k)$ provide a proper weighting scheme
for target distribution $\jointdens(\param,\latent_{1:T}) =
p(\param, \latent_{1:T}\mid \ssmobs_{1:T})$
(Definition \ref{def:proper}), for the following classes of functions,
respectively:
\begin{align*}
    \ref{item:filter-smoother}\quad
    \Ls & = L^1(\jointdens), &
    \ref{item:fwd-bwd-smoothing}\quad
    \Ls &= \big\{f\in L^1(\jointdens)\given f(\param,\latent_{1:T}) =
    f(\param,\latent_{t-1:t}),\text{ for some $t\in\{2{:}T\}$}\big\}, \\
    \ref{item:backwards-sampling}\quad
    \Ls & = L^1(\jointdens), &
    \ref{item:fwd-bwd-smoothing2}\quad
    \Ls &= \big\{f\in L^1(\jointdens)\given
    f(\param,\latent_{1:T}) =
    f(\param,\latent_{t}),\text{ for some $t\in\{1{:}T\}$}\big\}.
\end{align*}
In case $(\Param_k,U_k)_{k\ge 1}$ is a pseudo-marginal algorithm,
$\properweight_k \defeq \pr(\param_k) \properweightalt_k^{(i)}/U_k$.
\end{corollary} %}}}

%}}}

%%%%%%%%%%%%%%%%%%%%%%%%%%%%%%%%%%%%%%%%%%%%%%%%%%%%%%%%%%%%%%%%%%%%%%%%%%%%%%
\def\bibfont{\footnotesize}

\end{document}